\documentclass[12pt]{article}
\usepackage{geometry}
\usepackage{amsthm}
\usepackage{multirow}
\usepackage{amsmath}
\newtheorem{assumption}{Assumption}
\newtheorem{lemma}{Lemma}
\newtheorem{theorem}{Theorem}
\newtheorem{definition}{Definition}
\newtheorem{remark}{Remark}
\newtheorem{corollary}{Corollary}
\usepackage{amssymb, bbm}
\usepackage{graphicx}
\usepackage{enumerate}
\usepackage{natbib}
\usepackage[pdfencoding=auto,psdextra,unicode,hypertexnames=false]{hyperref}

\usepackage{booktabs}
\usepackage{textcomp} 
\usepackage{setspace}
\usepackage{float}
\usepackage{siunitx}
\usepackage{algorithm}
\usepackage{algpseudocode}

\title{\bf Aggregating Conformal Prediction Sets via $\alpha$-Allocation}
\author{Congbin Xu$^{*1}$,\ Yue Yu$^{*1}$,\ Haojie Ren$^{\text{†}2}$, Zhaojun Wang$^{1}$,\  and Changliang Zou$^{1}$              \\
	{$^1${\small\it School of Statistics and Data Science, Nankai University,Tianjin, China}} \\
	{$^2${\small\it{School of Mathematical Sciences, Shanghai Jiao Tong University, Shanghai, China}}}\\
}
\date{}
\def\alphabold{\boldsymbol{\alpha}}
\def\betabold{\boldsymbol{\beta}}
\def\vzero{\boldsymbol{0}}
\def\vone{\boldsymbol{1}}

\def\vt{\mathbf{t}}

\def\vv{\mathbf{v}}
\def\vw{\mathbf{w}}
\def\vx{\mathbf{x}}

\def\mW{\mathbf{W}}
\def\mX{\mathbf{X}}

\def\Dcal{{\mathcal{D}}}

\def\Gcal{{\mathcal{G}}}

\def\Ical{{\mathcal{I}}}

\def\Lcal{{\mathcal{L}}}

\def\Ncal{{\mathcal{N}}}

\def\Scal{{\mathcal{S}}}

\def\Xcal{{\mathcal{X}}}
\def\Ycal{{\mathcal{Y}}}

\def\Op{O_{p}}
\def\op{o_{p}}

\newcommand{\indicator}{\mathbbm{1}}
\def\normal{\mathcal{N}}
\def\P{\mathbb{P}}
\def\E{\mathbb{E}}

\def\real{\mathbb{R}}

\newcommand{\norm}[1]{\lVert #1\rVert}

\newcommand{\abs}[1]{|#1|}
\newcommand{\bigabs}[1]{\bigl| #1\bigr|}
\newcommand{\Bigabs}[1]{\Bigl| #1\Bigr|}

\newcommand{\lrs}[1]{\left(#1\right)}
\newcommand{\lrm}[1]{\left[#1\right]}
\newcommand{\lrl}[1]{\left\{#1\right\}}

\makeatletter

\newcommand{\Rmnum}[1]{{\rm\expandafter\@slowromancap\romannumeral #1@}}
\makeatother

\newcommand{\refeq}[1]{Eq.(\ref{#1})}

\begin{document}
\maketitle
\def\thefootnote{*}\footnotetext{The first two authors equally contributed to this work.}\def\thefootnote{\arabic{footnote}}
\def\thefootnote{†}\footnotetext{Corresponding author: \href{haojieren@sjtu.edu.cn}{haojieren@sjtu.edu.cn}}.
\begin{abstract}
Conformal prediction offers a distribution-free framework for constructing prediction sets with finite-sample coverage. Yet, efficiently leveraging multiple conformity scores to reduce prediction set size remains a major open challenge. Instead of selecting a single best score, this work introduces a principled aggregation strategy, COnfidence-Level Allocation (COLA), that optimally allocates confidence levels across multiple conformal prediction sets to minimize empirical set size while maintaining provable coverage. Two variants are further developed, COLA-s and COLA-f, which guarantee finite-sample marginal coverage via sample splitting and full conformalization, respectively. In addition, we develop COLA-l, an individualized allocation strategy that promotes local size efficiency while achieving asymptotic conditional coverage. Extensive experiments on synthetic and real-world datasets demonstrate that COLA achieves considerably smaller prediction sets than state-of-the-art baselines while maintaining valid coverage.
\end{abstract}
\section{Introduction}
Conformal prediction provides a general framework for constructing prediction sets with a finite-sample coverage guarantee, irrespective of the data distribution, the prediction model, or the choice of nonconformity score. Let $\Dcal=\lrl{\lrs{\mX_i,Y_i}}_{i=1}^n$ be a set of independent and identically distributed (i.i.d.) 
hold-out data from some joint distribution $P$, where $\mX_i \in\Xcal \subset\mathbb{R}^d$ represents features, and $Y_i \in\mathcal{Y}\subset\mathbb{R}$ is the response. For a test point $(\mX_{n+1},Y_{n+1})$ independently drawn from $P$, with $\mX_{n+1}$ observed but $Y_{n+1}$ unobserved, conformal prediction constructs a prediction set $\widehat{C}(\mX_{n+1};\alpha)$ such that $\P( Y_{n+1}\in\widehat{C}(\mX_{n+1};\alpha)) \geq 1-\alpha$
for a specified confidence level $1-\alpha\in(0,1)$. In the widely used split conformal prediction framework \citep{vovk2005algorithmic}, one assumes access to a nonconformity score function $S(\cdot,\cdot): \mathcal{X}\times\mathcal{Y}\rightarrow\mathbb{R}$, typically derived from a pretrained model, that quantifies the discrepancy between a hypothetical value $y$ and the model's prediction. The $(1-\alpha)$-level conformal prediction set for $Y_{n+1}$ then takes the form as $\widehat{C}(\mX_{n+1};\alpha) =\{ y : S(\mX_{n+1}, y) \leq Q_\alpha(\{S(\mX_i,Y_i)\}_{i= 1}^n\cup\{\infty\})\}$, where $Q_\alpha(\cdot)$ denotes the empirical $(1-\alpha)$ quantile.

In many applications, it is common to have access to multiple nonconformity score functions $\{S_k: \Xcal \times \Ycal \rightarrow \real\}_{k=1}^K$ with $K\geq 2$. Such scores may arise from predictive models trained on distinct data sources, implemented with different learning algorithms or hyperparameter configurations, or defined via heterogeneous nonconformity measures. The efficiency of a conformal prediction set--—typically measured by its size \citep{sadinle2019least,kiyani2024length}—--can vary markedly across scores. This raises a natural question: how to effectively leverage multiple scores to achieve smaller sets while maintaining the desired coverage? This problem has attracted considerable attention. One common strategy is to select a single best score from candidates by minimizing the set size \citep{yang2024selection,liang2024conformal}. However, there may be no clear winner among the candidate scores in complex real-world data. In such scenarios, an alternative strategy is to integrate information from multiple scores, possibly more efficient and stable than relying solely on a single selected one. The benefits of analogous principles have been supported in related domains, such as ensemble learning \citep{breiman1996bagging}, model averaging \citep{claeskens2008model}, and mixture of experts \citep{jacobs1991adaptive}. Motivated by this perspective, we aim to explore an aggregation-based approach to form conformal prediction sets. 

Recent works have begun to address this multi-score challenge, but important gaps remain. First, \textit{set-level combination} strategies \citep{wu2023multi,gasparin2024merging,qin2024sat} merge $K$ raw prediction sets at level $1-\alpha$. While these approaches guarantee (or approximately) marginal coverage under arbitrary dependence, they ignore the conformal machinery and therefore do not optimize the weights of the constituent scores to minimize the final set size. Second, \textit{score-level combination} methods instead attempt to form a single scalar score by a weighted average of the $S_k$'s or their underlying prediction models \citep{luo2025weighted,rivera2025conformal,yang2024selection}. This route, however, implicitly presumes different scores are directly commensurable. In practice, this assumption is often violated: residual-based scores, quantile-based scores \citep{romano2019conformalized}, and other heterogeneous scores may differ both in scale and interpretation.

This paper introduces a new method that is closely related to, but distinct from, existing approaches previously discussed. Our approach integrates information from multiple scores by allocating different confidence levels across the sets induced by candidate score functions, and then intersecting these sets to form the final prediction set. Although the general idea of intersecting level-adjusted conformal prediction sets traces back to Bonferroni-type constructions (e.g., \cite{lei2018distribution}), our contribution lies in constructing data-driven confidence level allocations. We formalize this idea in a framework of COnfidence-Level Allocation (COLA), which adaptively aggregates conformal prediction sets via optimally assigned confidence levels.

The core idea of the proposed method is outlined here, with precise definitions and notation deferred to Section~\ref{sec:COLA-marginal}. Let $\widehat{C}_{k}(\vx;\alpha_k)=\lrl{y:S_k(\vx,y)\leq Q_{\alpha_k} (\{ S_{k,i} \}_{i = 1}^n \cup \{\infty\}}$ denote the conformal prediction set induced by score function $S_k$, where $S_{k,i} = S_k(\mX_i,Y_i)$ for brevity.
We treat the overall confidence level $\alpha$ as a budget, allocated across the $K$ candidate sets to minimize the average size of the aggregated prediction set. Specifically, the allocation $\widehat{\alphabold}=\lrs{\widehat{\alpha}_1,\dots,\widehat{\alpha}_K}$ is obtained by solving the optimization problem:
$$\widehat{\alphabold}\in \arg \min_{\alphabold\in \boldsymbol{\Theta}} \frac{1}{n} \sum_{i=1}^n \Big| \cap_{k = 1}^K \widehat{C}_{k} (\mX_{i};\alpha_k) \Big|, $$where $ \boldsymbol{\Theta} = \{\alphabold \in \real^K : \norm{\alphabold}_1 = \alpha,\alphabold\geq\vzero\}$. The aggregated conformal prediction set is accordingly constructed as
\begin{align}\label{cpso}
	\widehat{C}(\mX_{n+1};\alpha) = \cap_{k = 1}^K \widehat{C}_{k} (\mX_{n+1};\widehat{\alpha}_{k}).
\end{align}
We say that a conformal prediction set satisfies asymptotically optimal allocation if it is asymptotically valid and its expected size converges to that of the oracle allocation in the population construction (see Definition~\ref{def:asy_opt}). Under mild conditions, the set in \refeq{cpso} attains this optimality.
This procedure prioritizes set efficiency and is thus referred to as $\mathrm{COLA\text{-}e}$. As it fails to guarantee finite-sample validity, we also consider two variants: $\mathrm{COLA\text{-}s}$, which employs sample splitting and is straightforward to implement, and $\mathrm{COLA\text{-}f}$, which adopts a full-conformal strategy to enhance efficiency at the expense of higher computational cost.

Additionally, existing approaches for score selection or aggregation have primarily targeted at improving average size efficiency, failing to adapt choices to individual test points. Individualized procedures are increasingly crucial, particularly in contexts of personalized adaptation such as precision medicine, where predictive information from multiple hospitals—--each with varying patient demographics and diagnostic equipment—--must be carefully aggregated to guide treatment decisions for a specific patient. Motivated by this need, we propose $\mathrm{COLA\text{-}l}$, an individualized allocation strategy that effectively utilizes multiple score functions to achieve size efficiency at the individual level while guaranteeing asymptotic conditional coverage \citep{lei2014distribution,guan2022localized}.

The main contributions of this work are summarized as follows:
\begin{itemize}
	\item We introduce COLA, a general framework for aggregating multiple conformal prediction sets through data-driven confidence-level allocation, yielding more efficient prediction regions.
	\item Four principled variants are developed: $\mathrm{COLA\text{-}e}$, which emphasizes efficiency; $\mathrm{COLA\text{-}s}$ and $\mathrm{COLA\text{-}f}$, which ensure finite-sample marginal validity; and $\mathrm{COLA\text{-}l}$, which achieves individualized efficiency.
	\item Under mild conditions, we show COLA satisfies asymptotically optimal allocation in both marginal and individualized regimes.
	\item Extensive empirical experiments on synthetic and real-world datasets corroborate the theory and demonstrate the flexibility and performance advantages of COLA.
\end{itemize}

\subsection{Related Work}\label{subsec:related work}
We begin by elaborating on the score selection methods. \cite{yang2024selection} proposed two selection algorithms: one achieves finite-sample efficiency but only asymptotic validity, while the other uses sample splitting to guarantee finite-sample validity at the cost of reduced efficiency. \cite{liang2024conformal} further improved upon these methods by leveraging the full conformal prediction procedure. These approaches are easy to implement and supported by comprehensive theoretical guarantees, making them attractive in practice. However, restricting attention to a single score from the candidates may overlook important information, particularly when candidate scores capture different aspects of the data or when none is well aligned with the data distribution.

Set-level combination strategies aim to merge $K$ individual prediction sets, each calibrated at level $1 - \alpha$, possibly under arbitrary dependence. \cite{gasparin2024merging} proposed a majority vote rule that includes labels covered by more than half of the sets, achieving marginal coverage of at least $1-2\alpha$. Other works, such as \citet{qin2024sat}, \citet{wu2023multi} and \citet{wong2025improving}, considered merging multiple conformal p-values for a hypothesized label $y$, retaining labels whose aggregated p-values exceed $\alpha$. While these methods guarantee valid coverage, they do not explicitly target optimizing the efficiency of conformal prediction sets.

In addition to ensuring coverage, a key objective in conformal prediction is achieving set efficiency, typically in terms of set size. Theoretically, size optimality has been established under structural assumptions such as consistent quantile regression~\citep{sesia2020comparison}, symmetric noise~\citep{lei2018distribution}, and other specific settings~\citep{burnaev2014efficiency,izbicki2019flexible}. Beyond these, recent methods have pursued efficiency improvements in more general settings, either by refining scores~\citep{xie2024boosted}, incorporating conformal training objectives ~\citep{bellotti2021optimized,einbinder2022training,stutz2022learning}, or directly minimizing the set size through a constrained optimization problem~\citep{bai2022efficient,fan2023utopia,kiyani2024length}. However, a unified, data-driven framework for efficiently aggregating multiple scores has not been provided, which motivates our proposed COLA framework.

\subsection{Organization and Notation}
\textbf{Organization.} The remainder of the paper is organized as follows. Section~\ref{sec:COLA-marginal} introduces our proposed level allocation method for combining multiple scores. 
Section~\ref{sec:COLA-conditional} extends this framework to an individualized
allocation strategy while ensuring asymptotic conditional coverage. Section~\ref{sec:numerical_result} presents empirical comparisons with existing methods. Finally, Section~\ref{sec:conclusion} concludes the paper.

\noindent\textbf{Notations.} 
For any positive integer $m$, define $[m]=\{1,\dots,m\}$. For any real number $a \in \real$, $|a|$ denotes its absolute value. For a discrete set $A$, $|A|$ denotes its cardinality; for a continuous set $A \subset \mathbb{R}$, $|A|$ denotes its Lebesgue measure. Bold letters (e.g., $\vv$) denote vectors, and nonbold symbols (e.g., $v_i$) represent their components. For any $\vv \in \real^m$, the support of $\vv$ is defined as $\mathrm{supp}(\vv) = \{i \in [m]: v_i \neq 0\}$ and $\norm{\vv}$ represents its $\ell_2$-norm. $\delta_x$ denotes the Dirac distribution at $x$. For an event $A$, $\indicator\{A\}$ denotes its indicator. $\E_{\mX_{n+1}} [\cdot]$ denotes the expectation with respect to $\mX_{n+1}$.

\section{Aggregation via \texorpdfstring{$\alpha$}{Alpha}-Allocation}\label{sec:COLA-marginal}

This section formalizes the construction of efficient conformal prediction sets using multiple fixed, pre-trained nonconformity score functions $\{S_k: \Xcal \times \Ycal \rightarrow \real\}_{k \in [K]}$. 

Recall that the hold-out data $\Dcal = \{(\mX_i,Y_i): i \in [n]\}$ and the test point $\mX_{n+1}$ are drawn i.i.d. from the same distribution. 
Building on the intersecting level-adjusted sets framework, for any fixed feasible allocation $\alphabold=(\alpha_1,\dots,\alpha_K) \in \boldsymbol{\Theta}$, the aggregated conformal prediction set $\cap_{k = 1}^K \widehat{C}_k(\mX_{n+1};\alpha_k)$ satisfies 
\begin{align}\label{eq:intersection_coverage}
	\P(Y_{n+1} \in \cap_{k = 1}^K \widehat{C}_k(\mX_{n+1};\alpha_k) ) \geq 1 - \sum_{k\in[K]} \P(Y_{n+1} \notin \widehat{C}_k(\mX_{n+1};\alpha_k)) \geq 1 - \alpha,
\end{align}
provided that each prediction set $\widehat{C}_k(\mX_{n+1};\alpha_k)$ controls marginal coverage at level $1 -  \alpha_k$, as ensured by standard conformal prediction. This naturally raises two key questions: what constitutes the optimal allocation, and how can it be estimated from data? Moreover, how can valid coverage be maintained when the allocation depends on the data? These questions are addressed in the following sections.

\subsection{Oracle Confidence Allocation}

We start by defining the oracle allocation as the choice of $\alphabold$ that minimizes the expected size of the aggregated prediction set while satisfying the coverage guarantee at the population level. For each score function $S_k$, the oracle prediction set achieving $1-\alpha_k$ coverage is
\begin{align*}
	C_k(\vx;\alpha_k) = \{y: S_k(\vx,y) \leq q_{k,\alpha_k}\},
\end{align*}
where $q_{k,\alpha}$ denotes the $(1-\alpha)$ quantile of the distribution of $S_k(\mX_{n+1},Y_{n+1})$. 
Given any feasible allocation $\alphabold \in \boldsymbol{\Theta}$, it follows from the definition of the quantile function and \refeq{eq:intersection_coverage} that the intersection $\cap_{k = 1}^K C_k(\mX_{n+1};\alpha_k)$ attains valid marginal coverage. The efficiency of a specific allocation $\alphabold$ is measured by the expected size of the aggregated set,
\begin{align}\label{eq:population_loss}
	\Lcal(\alphabold) & = \E_{\mX_{n+1}}[| \cap_{k = 1}^K C_k(\mX_{n+1};\alpha_k) |].
\end{align}
The oracle allocation is then defined as any minimizer 
$$\alphabold^* \in\arg\min_{\alphabold \in \boldsymbol{\Theta}} \Lcal(\alphabold).$$ 
This criterion coincides with the principle of ambiguity minimization in the multiclass classification problem \citep{sadinle2019least}. The oracle allocation $\alphabold^*$ depends on both the pretrained score functions $\{S_k\}_{k \in [K]}$ and the underlying data distribution. Intuitively, when a score function captures the covariate-response relationship more effectively, its associated prediction set tends to be smaller, and the oracle allocation $\alphabold^*$ assigns it a larger share of the confidence budget. In this sense, $\alphabold^*$ reflects the relative importance of different score functions.

It is worth noting that  $\alphabold^*$ may not be unique, and different oracle allocations $\alphabold^*$ may lead to different intersections $\cap_{k = 1}^K C_k(\mX_{n+1};\alpha_k^*)$. Our focus, however, is on the size $\Lcal(\alphabold^*)$, which depends on the score functions and the underlying data distribution. Throughout, we treat the score functions as fixed and study the optimal value attainable under these given scores.

A desired prediction set is one that satisfies the coverage requirement while achieving an expected size asymptotically equal to $\Lcal(\alphabold^*)$, formally defined as follows.
\begin{definition}[Asymptotically optimal allocation]\label{def:asy_opt}
	A conformal prediction set $\widehat{C}(\mX_{n+1};\alpha) $ satisfies asymptotically optimal allocation if it is asymptotically valid, 
	$$\P(Y_{n+1} \in \widehat{C}(\mX_{n+1};\alpha)) \geq 1 - \alpha + o(1),$$ 
	and its expected size converges to the oracle value,
	\[\E_{\mX_{n+1}}[\abs{\widehat{C}(\mX_{n+1};\alpha)}] - \Lcal(\alphabold^*) = \op(1).\]
\end{definition}
Similar notions of optimality are commonly used in the literature on model selection and model averaging \citep{li1987asymptotic,hansen2007least}.

\subsection{COLA-e: Aggregation via Confidence Allocation}
Our empirical approach approximates the allocation by replacing $C_k(\vx;\alpha_k)$ with their conformal prediction sets in \refeq{eq:population_loss} and performing empirical risk minimization based on the hold-out data $\Dcal$.

For each score function $S_k$, the oracle prediction set $C_k(\vx;\alpha_k)$ is estimated by the conformal prediction set
\begin{align*}
	\widehat{C}_{k} (\vx;\alpha_k) & = \{y: S_k(\vx,y) \leq Q_{\alpha_k}(\lrl{S_{k,i}}_{i=1}^n\cup\{\infty\}) \}.
\end{align*}
The allocation $\alphabold^*$ is then estimated by solving the following empirical risk minimization problem:
\begin{equation}\label{eq:alpha_eff}
	\begin{aligned}
		\Lcal_n(\alphabold) &= \frac{1}{n} \sum_{i \in [n]} | \cap_{k = 1}^K \widehat{C}_{k} (\mX_{i};\alpha_k) |,\\
		\widehat{\alphabold}^{\mathrm{e}} & \in \arg \min_{\alphabold \in \boldsymbol{\Theta}} \Lcal_n(\alphabold).
	\end{aligned}
\end{equation}
The resulting aggregated prediction set, referred to as $\mathrm{COLA\text{-}e}$, is
\[\widehat{C}^{\mathrm{e}}(\mX_{n+1};\alpha) = \cap_{k = 1}^K \widehat{C}_{k} (\mX_{n+1};\widehat{\alpha}_{k}^\mathrm{e}) .\]

$\mathrm{COLA\text{-}e}$ is closely related to Efficiency First Conformal Prediction (EFCP) proposed by \cite{yang2024selection}, but it operates within a more general framework. In particular, if the feasible region $\boldsymbol{\Theta}$ is restricted to $\boldsymbol{\Theta}_{\mathrm{res}} = \{\alphabold \in \boldsymbol{\Theta} : \abs{\mathrm{supp}(\alphabold)} = 1\}$, then $\mathrm{COLA\text{-}e}$ reduces to EFCP. This demonstrates that $\mathrm{COLA\text{-}e}$ can produce smaller prediction sets than EFCP, particularly when multiple scores capture complementary aspects of the response. As illustrated in Figure~\ref{fig:intuition} of the Supplementary Material, it may occur that while each individual conformal prediction set becomes slightly larger after $\alpha$ allocation, their intersection is notably shorter than any single $(1-\alpha)$-level set.

We now examine the theoretical properties of $\mathrm{COLA\text{-}e}$. For clarity of exposition, the main text focuses on the regression setting. Extensions to the classification setting are provided in Section~\ref{sec:classification} of the Supplementary Material. Since the allocation $\widehat{\alphabold}^{\mathrm{e}}$ is estimated from $\Dcal$, it introduces a data-dependent asymmetry: the test point plays no role in the construction of $\widehat{\alphabold}^{\mathrm{e}}$, thus breaking exchangeability. Therefore, the finite-sample coverage guarantee of COLA-e no longer holds. The following theorem establishes its asymptotic validity.
\begin{theorem}[Validity of $\mathrm{COLA\text{-}e}$]\label{thm:valid_cola_e} 
	The $\mathrm{COLA\text{-}e}$ prediction set $\widehat{C}^{\mathrm{e}}(\mX_{n+1};\alpha)$ satisfies
	$$\P(Y_{n+1}\in \widehat{C}^{\mathrm{e}}(\mX_{n+1};\alpha))\geq 1-\alpha-O\left(\sqrt{\frac{\log(Kn)}{n}}\right).$$ 
\end{theorem}
Theorem~\ref{thm:valid_cola_e} requires only that the observations $\lrl{(\mX_i,Y_i)}_{i\in[n+1]}$ are i.i.d., but it achieves a slower coverage convergence rate of order $\sqrt{\log(Kn)/n}$ compared to EFCP \citep{yang2024selection}, whose rate is $\sqrt{\log(K)/n}$. This slight degradation reflects the increased complexity of the allocation space $\boldsymbol{\Theta}$ relative to EFCP’s restricted space $\boldsymbol{\Theta}_{\mathrm{res}}$. The proof leverages the multivariate Dvoretzky–Kiefer–Wolfowitz inequality \citep{naaman2021tight}. 

The following assumptions are required to ensure the efficiency of the COLA-e.
\begin{assumption}\label{ass:marginal_ass}
	There exist constants $l_\alpha$, $L_q$ and $M$ such that: 
	\begin{itemize}
		\item[(a).] The distribution functions $F_k(t) = \P(S_k(\mX_{n+1},Y_{n+1}) \leq t)$ are continuous.
		\item[(b).] The quantile functions $\{q_{k,\alpha}\}_{k = 1}^K$ are $L_q$-Lipschitz continuous on $[ l_\alpha - \xi , \alpha +\xi ]$ with $\xi \geq \sqrt{\log(2K)/(2n)} + (2-l_\alpha)/n$, such that $\abs{q_{k,\alpha_1} - q_{k,\alpha_2}} \leq L_q \abs{\alpha_1 - \alpha_2}$ for any $\alpha_1$, $\alpha_2 \in [ l_\alpha - \xi , \alpha +\xi ]$.
		\item[(c).] For any $\alphabold = (\alpha_k)_{k \in [K]} \in \boldsymbol{\Theta}$ and $\alpha_k^\prime = \max(\alpha_k, l_\alpha)$, it holds that		$\cap_{k = 1}^K C_k(\vx;\alpha_k) = \cap_{k = 1}^K C_k(\vx;\alpha_k^\prime)$ and $\cap_{k = 1}^K \widehat{C}_{k}(\vx;\alpha_k)  = \cap_{k = 1}^K \widehat{C}_{k}(\vx;\alpha_k^\prime).$
		
		\item[(d).] The size of the intersection is uniformly bounded: $\sup_{\alphabold\in\boldsymbol{\Theta}, \vx}\abs{\cap_{k=1}^KC_k(\vx;\alpha_k)} \leq M.$
	\end{itemize}
\end{assumption}

Assumption~\ref{ass:marginal_ass}(a) is standard and has been adopted in prior works (see \citep{liang2024conformal}). Assumption~\ref{ass:marginal_ass}(b) ensures that the empirical quantile function $Q_\alpha(\{S_{k,i}\}_{i= 1}^n\cup\{\infty\})$ converges uniformly to its population counterpart $q_{k,\alpha}$, which holds if the density function $f_k(t)$ of $F_k(t)$ is uniformly bounded away from zero for $t \in [q_{k, \alpha + \xi}, q_{k, l_\alpha - \xi}]$. Similar assumptions can be found in \cite{lei2018distribution} and \cite{yang2024selection}. Assumption~\ref{ass:marginal_ass}(c) implies that very small allocations $\alpha_k$ (i.e., $\alpha_k \leq l_\alpha$) yield noninformative prediction sets $C_k(\vx;\alpha_k)$ and  $\widehat{C}_k(\vx;\alpha_k)$, so truncating $\alpha_k$ to $l_\alpha$ does not affect the final prediction sets. This aligns with the observation that $\alphabold^*$ or $\widehat{\alphabold}^{\mathrm{e}}$ typically have sparse support, indicating that only a small subset of prediction sets contributes significantly. Assumption~\ref{ass:marginal_ass}(d) states that the size of the intersection $\cap_{k = 1}^K C_k(\vx;\alpha_k)$ is uniformly bounded.
\begin{assumption}\label{ass:common_ass}
	The set $\{y: S_k(\vx,y) \leq t\}$ satisfies: 
	\begin{itemize}
		\item [(a).] For any $t \in \real$ and $\vx \in \Xcal$,  $\{y: S_k(\vx,y) \leq t\}$ can be represented as a union of at most $m$ mutually disjoint intervals. 
		\item [(b).] There exists a constant $L_s > 0$ such that for any score $S_k$ and $h>0$, 
		\[\sup_{t,\vx} [ \abs{\{y: S_k(\vx,y) \leq t + h\}} - \abs{\{y: S_k(\vx,y) \leq t\}}]\leq L_s h . \]
	\end{itemize}
\end{assumption}
Assumption~\ref{ass:common_ass} is satisfied by many commonly-used scores in regression settings. For instance, the residual score $\abs{y - \widehat{\mu}(\vx)}$ yields $m = L_s = 1$; the conformalized quantile score $\max\lrl{\widehat{\tau}_{\alpha/2}(\vx)-y,y-\widehat{\tau}_{1 - \alpha/2}\lrs{\vx}}$ yields $m = L_s = 1$, where $\widehat{\tau}_{\alpha/2}(\vx)$ and $\widehat{\tau}_{1 - \alpha/2}\lrs{\vx}$ are pretrained estimates of the conditional quantiles of $Y_{n+1}$ given $X_{n+1} = \vx$ at levels $\alpha/2$ and $(1-\alpha/2)$, respectively; and the rescaled residual score $\abs{y - \widehat{\mu}(\vx)}/\widehat{\sigma}(\vx)$ satisfies the assumption with $m = 1$ and $L_s = \sup_{\vx} \widehat{\sigma}(\vx)$ provided that the scaling function $\widehat{\sigma}(\vx)$ is uniformly bounded. Similar conditions are also adopted in \cite{yang2024selection} to derive the convergence of the prediction set from the convergence of the quantile function. However, Assupmtion~\ref{ass:common_ass} dose not hold in classification problems. In Appendix, we present general theoretical guarantees under an alternative formulation that does not rely on this assumption.
\begin{theorem}\label{thm:asymptotic_opt_cola_e}
	Under Assumptions~\ref{ass:marginal_ass} and \ref{ass:common_ass},
	the $\mathrm{COLA\text{-}e}$ prediction set satisfies
	\[\E_{\mX_{n+1}} [\abs{\widehat{C}^{\mathrm{e}}(\mX_{n+1};\alpha)}] - \Lcal(\alphabold^*) = \Op\left( m L_q L_s \sqrt{\frac{\log(K)}{n}}  + M \sqrt{\frac{K \log(n)}{n}}\right).\]
\end{theorem}
This theorem demonstrates that when $K = o(n/\log(n))$, the expected size of the $\mathrm{COLA\text{-}e}$ prediction set converges to that of the oracle set. Taken together, Theorems~\ref{thm:valid_cola_e} and \ref{thm:asymptotic_opt_cola_e} show $\mathrm{COLA\text{-}e}$ satisfies asymptotically optimal allocation.

\subsection{COLA-s: Finite-sample Validity via Sample Splitting}
Since the allocation $\widehat{\alphabold}^{\mathrm{e}}$ in $\mathrm{COLA\text{-}e}$ is estimated from the same data used for calibration, the resulting prediction set $\widehat{C}^{\mathrm{e}}(\mX_{n+1};\alpha)$ does not guarantee finite-sample validity. To address this, we introduce a sample-splitting variant, $\mathrm{COLA\text{-}s}$, which parallels the Validity First Conformal Prediction (VFCP) approach of \cite{yang2024selection}.

Specifically, the hold-out data $\Dcal$ is randomly partitioned into two disjoint subsets $\Dcal_{\mathrm{tu}}$ and $\Dcal_{\mathrm{cal}}$, with indices $\Ical_{\mathrm{tu}}$ and $\Ical_{\mathrm{cal}}$ and sizes $n_{\mathrm{tu}}$ and $n_{\mathrm{cal}}$, respectively. The first part $\Dcal_{\mathrm{tu}}$ is used to estimate the allocation, while the second subset $\Dcal_{\mathrm{cal}}$ is used as an independent sample set for calibration.

The allocation is learned on $\Dcal_{\mathrm{tu}}$ by solving 
\begin{equation*}
	\widehat{\alphabold}^{\mathrm{s}} \in \arg \min_{\alphabold\in \boldsymbol{\Theta}} \frac{1}{n_{\mathrm{tu}}} \sum_{i \in \Ical_{\mathrm{tu}}} | \cap_{k = 1}^K \widehat{C}_{k,\mathrm{tu}} (\mX_{i};\alpha_k) |,
\end{equation*}
where $\widehat{C}_{k,\mathrm{tu}} (\vx;\alpha_k) = \{y: S_k(\vx,y) \leq Q_{\alpha_k} (\{ S_{k,i} \}_{i \in \Ical_{\mathrm{tu}}} \cup \{\infty\}) \}$.
The $\mathrm{COLA\text{-}s}$ prediction set is then constructed on $\Dcal_{\mathrm{cal}}$ as
\[\widehat{C}^{\mathrm{s}}(\mX_{n+1};\alpha) = \cap_{k = 1}^K \widehat{C}_{k,\mathrm{cal}} (\mX_{n+1};\widehat{\alpha}_{k}^{\mathrm{s}}),\]
with 
$\widehat{C}_{k,\mathrm{cal}} (\vx;{\alpha}_{k})  = \{y: S_k(\vx,y) \leq Q_{{\alpha}_{k}} (\{ S_{k,i} \}_{i \in \Ical_{\mathrm{cal}}} \cup \{\infty\}) \}$.

The following theorem shows that $\mathrm{COLA\text{-}s}$ guarantees exact finite-sample coverage under exchangeability. 
\begin{theorem}[Finite-sample validity of $\mathrm{COLA\text{-}s}$]\label{thm:cola_s_coverage}
	If the calibration data and the test point $\{(\mX_i,Y_i): i \in \Ical_{\mathrm{cal}} \cup \{n+1\}\}$ are exchangeable, then the $\mathrm{COLA\text{-}s}$ prediction set satisfies 
	$\P(Y_{n+1}\in\widehat{C}^{\mathrm{s}}\lrs{\mX_{n+1}; \alpha}) \geq 1-\alpha.$
\end{theorem}

The following corollary, a direct consequence of Theorem~\ref{thm:asymptotic_opt_cola_e}, shows that $\mathrm{COLA\text{-}s}$ also attains asymptotically optimal allocation.
\begin{corollary}\label{thm:optimal_cola_s}
	Under Assumptions~\ref{ass:marginal_ass} and \ref{ass:common_ass} (with Assumption~\ref{ass:marginal_ass}(c) applied separately to  $\widehat{C}_{k,\mathrm{tu}}$ and $\widehat{C}_{k,\mathrm{cal}}$), 
	the $\mathrm{COLA\text{-}s}$ prediction set satisfies
	\[\E_{\mX_{n+1}} [\abs{\widehat{C}^{\mathrm{s}}(\mX_{n+1};\alpha)}] - \Lcal(\alphabold^*) = \Op\left( m L_q L_s \sqrt{\frac{\log(K)}{n_{\mathrm{min}}}} + M \sqrt{\frac{K \log(n_{\mathrm{min}})}{n_{\mathrm{min}}}}\right) ,\]
	where $n_{\mathrm{min}} = \min(n_{\mathrm{tu}}, n_{\mathrm{cal}})$.
\end{corollary}
This result illustrates that $\mathrm{COLA\text{-}s}$ guarantees exact finite-sample coverage, but the use of sample-splitting reduces efficiency.

\subsection{COLA-f: Finite-sample Validity via Full Conformalization}
As discussed above, $\mathrm{COLA\text{-}e}$ lacks finite-sample coverage guarantees, while $\mathrm{COLA\text{-}s}$ suffers from efficiency loss due to sample splitting. A potential remedy for both issues is to adopt a full conformal strategy \citep{vovk2005algorithmic,liang2024conformal}, which restores exchangeability by  treating the hold-out data and the test point symmetrically.

Specifically, given a hypothesized value $y$ for $Y_{n+1}$, the test-augmented allocation is as follows:
\begin{align*}
	\widehat{\alphabold}^{(y)}\in\arg\min_{\alphabold \in \boldsymbol{\Theta}} \frac{1}{n+1}\sum_{i\in[n+1]} | \cap_{k=1}^K\widehat{C}_{k}^{(y)}\lrs{\mX_i;\alpha_k}|,
\end{align*}
where $\widehat{C}_{k}^{(y)} (\vx;\alpha_k) 
= \{\tilde{y}: S_k\lrs{\vx,\tilde{y}} \leq Q_{\alpha_k} (\{ S_{k,i} \}_{i = 1}^n \cup \{S_k(\mX_{n+1},y)\}) \}$.
Intuitively, $\widehat{\alphabold}^{(y)}$ identifies the allocation that minimizes the average prediction set size when all $n+1$ data points, including the hypothesized test pair $(\mX_{n+1}, y)$, are used simultaneously for calibration and evaluation, thereby restoring symmetry across the dataset. 

The $\mathrm{COLA\text{-}f}$ prediction set is then defined as:
$$\widehat{C}^{\mathrm{f}}(\mX_{n+1};\alpha)=\{y:y\in\cap_{k=1}^K\widehat{C}_{k}^{(y)}(\mX_{n+1},\widehat{\alpha}_k^{(y)})\}.$$

When the feasible region $\boldsymbol{\Theta}$ is restricted to $\boldsymbol{\Theta}_{\mathrm{res}} = \{\alphabold \in \boldsymbol{\Theta} : \abs{\mathrm{supp}(\alphabold)} = 1\}$, $\mathrm{COLA\text{-}f}$ coincides with the method proposed in \cite{liang2024conformal}. The following theorem establishes the finite-sample coverage of $\widehat{C}^{\mathrm{f}}(\mX_{n+1};\alpha)$ under exchangeability.

\begin{theorem}[Finite-sample validity of $\mathrm{COLA\text{-}f}$]\label{thm:valid_full}
	Suppose that $\{(\mX_i, Y_i)\}_{i=1}^{n+1}$ are exchangeable. Then, $\mathrm{COLA\text{-}f}$ satisfies that
	$\P(Y_{n+1}\in \widehat{C}^{\mathrm{f}}(\mX_{n+1};\alpha) )\geq 1-\alpha.$
\end{theorem}

Despite ensuring finite-sample coverage, $\mathrm{COLA\text{-}f}$ has two drawbacks. First, it requires an exhaustive search over the entire label space, leading to high computational cost, especially in regression settings. Second, the complex form of $\widehat{C}^{\mathrm{f}}$ makes it difficult to establish theoretical guarantees for size efficiency.  
Empirical results in Section~\ref{subsec:comparison_s_e_f} of the Supplementary Material show that $\mathrm{COLA\text{-}f}$ often yields smaller prediction sets when the sample size is small, but at the expense of substantially higher computational cost than $\mathrm{COLA\text{-}e}$ and $\mathrm{COLA\text{-}s}$. 

\subsection{Stepwise Optimization}

Here, we turn to the practical implementation, focusing on how to efficiently solve the empirical risk minimization problem of $\mathrm{COLA\text{-}e}$. The problem in \refeq{eq:alpha_eff} can be formulated as a linearly constrained optimization problem:
\begin{align*}
	\min_{\alpha_k} \sum_{i\in[n]} \abs{\cap_{k = 1}^K \widehat{C}_k(\mX_i; \alpha_k)}  ~~~~ \text{s.t. } \alpha_k \geq 0 \text{ and } \sum_{k\in[K]} \alpha_k = \alpha.
\end{align*}
Since each $\widehat{C}_k(\vx;\alpha_k)$ changes only at the discrete grid  $\Gcal(n)=\lrl{0,1/n,2/n,\dots,1}$, the overall loss function $\Lcal_n(\alphabold)$ is piecewise constant, making the problem inherently a combinatorial optimization task \citep{korte2008combinatorial}.

To address this issue, we can adopt a grid search strategy over the discretized grid $\boldsymbol{\Theta}_{\mathrm{grid}}(\alpha,n) = \{\alphabold \in \boldsymbol{\Theta}: \alpha_k \in \Gcal(n)\}$ to identify the empirical minimizer $\widehat{\alphabold}$. This naive grid search requires computational complexity $O((\alpha n)^{K-1})$, which is impractical for large $K$. As stated in Assumption~\ref{ass:marginal_ass}(c), the optimal allocation is typically sparse, which implies that one needs to assign nonzero shares of the total confidence budget to only a small subset of prediction sets.

To exploit this sparsity and reduce computational cost, we instead adopt a stepwise optimization algorithm inspired by classical variable selection methods \citep{claeskens2006variable}. At each iteration, the stepwise algorithm performs two operations: a forward step that identifies the score yielding the largest decrease in prediction set size, and a backward step that removes redundant scores.

\textbf{Forward step.} Given the current selected set $\Scal^{(t-1)}$ (with $\Scal^{(0)} = \varnothing$), we check whether adding a new score index $k \in [K]$ leads to a smaller empirical loss 
\[\min_{\substack{\alphabold \in \boldsymbol{\Theta}_{\mathrm{grid}}(\alpha,n)\\ \mathrm{supp}(\alphabold) \subset \Scal^{(t-1)} \cup \{k\}}} \Lcal_n(\alphabold) < \min_{\substack{\alphabold \in \boldsymbol{\Theta}_{\mathrm{grid}}(\alpha,n)\\ \mathrm{supp}(\alphabold) \subset \Scal^{(t-1)}}} \Lcal_n(\alphabold).\] 
If so, we add the index $k^\prime$ that achieves the largest reduction, updating $\Scal^{(t)} = \Scal^{(t-1)} \cup \{k^\prime\}$. If no such $k^\prime$ exists, the algorithm stops and solves the optimization problem subject to the support constraint $\mathrm{supp}(\alphabold) \subset \Scal^{(t-1)}$.

\textbf{Backward step.} After the forward step, we re-optimize with the set $\Scal^{(t)}$ 
\[\widehat{\alphabold}^{(t)} = \arg \min_{\substack{\alphabold \in \boldsymbol{\Theta}_{\mathrm{grid}}(\alpha,n)\\ \mathrm{supp}(\alphabold) \subset \Scal^{(t)}}} \Lcal_n(\alphabold).\]
Any selected index $k \in \Scal^{(t)}$ with $\widehat{\alpha}^{(t)}_k = 0$ is eliminated, and we update $\Scal^{(t)} = \mathrm{supp}(\widehat{\alphabold}^{(t)}).$

The process continues until the set size $\abs{\Scal^{(t)}}$
exceeds a pre-specified $k_{\max}$ or a maximum number of iterations is reached. This reduces complexity to $O(K (\alpha n)^{k_{\max}-1})$, making the procedure scalable. This algorithm is summarized in Algorithm~\ref{alg:stepwise}. Empirical results in Section~\ref{sec:optimization_comparision} of the Supplementary Material show that stepwise optimization outperforms smoothing-based relaxations \citep{xie2024boosted,kiyani2024length}, particularly when $K$ is large, by producing smaller sets while avoiding the pitfalls of smooth approximations such as parameter tuning, approximation error, and poor local minima. 

\section{Individualized Level Allocation}
\label{sec:COLA-conditional}

While Section~\ref{sec:COLA-marginal} focused on the average-case allocation that minimizes the expected prediction set size, many downstream tasks demand decisions tailored to the individual test point $\mX_{n+1}$. 
To meet this need, we consider individualized confidence-level allocation, which learns a vector $\alphabold(\mX_{n+1})$ depending on the test point $\mX_{n+1}$. The primary objective is to minimize the size of the aggregated intersection set specifically at $\mX_{n+1}$. The key enabler for such local optimization is a conditional coverage guarantee
\begin{align*}
\P(Y_{n+1} \in \widehat{C}(\mX_{n+1};\alpha) | \mX_{n+1}) \geq 1 - \alpha,
\end{align*}
which is stronger than the marginal guarantee considered earlier. However, finite-sample conditional coverage is impossible in a distribution-free setting \citep{foygel2021limits}. This motivates the study on asymptotic conditional validity, a line of research that has received growing attention in recent years \citep{guan2022localized,chernozhukov2021distributional,gibbs2025conformal}.

Following this perspective, we 
construct each prediction set via localized conformal prediction \citep{guan2022localized}.
Concretely, we employ kernel-weighted quantiles to build $K$ sets $\widehat{C}_{k}^{\mathrm{loc}} (\mX_{n+1};\alpha_k)$ that are (asymptotically) conditionally valid at $\mX_{n+1}$. Within this framework, the individual allocation $\alphabold(\mX_{n+1})$ is then optimized exactly as before, but now targeting the conditional set size rather than the global expectation. This leads to a finer-grained decision rule, which we refer to as $\mathrm{COLA\text{-}l}$.

\subsection{COLA-l: COLA with Individualized Allocation}

Equipped with a kernel function $H: \Xcal \times \Xcal \rightarrow \real$ and a score $S_k$, the localized conformal prediction set \citep{guan2022localized} is defined as
\begin{align*}
\widehat{C}_{k}^{\mathrm{loc}} (\mX_{n+1};\alpha_k) = \{y: S_k(\mX_{n+1},y) \leq Q_{\alpha_k} (\{ S_{k,i} \}_{i = 1}^n , \vw) \},
\end{align*}
where $Q_\alpha(\{ S_{k,i} \}_{i = 1}^n, \vw)$ denotes the $(1-\alpha)$ quantile of the weighted empirical distribution $\sum_{i\in [n]} w_i \delta_{S^k_i} $ with weights $\vw = (w_1,\dots,w_n)$ defined by $w_i = H(\mX_i, \mX_{n+1})/\sum_{j\in[n]} H(\mX_j, \mX_{n+1})$. The weights capture the similarity between $\mX_i$ and the test point $\mX_{n+1}$, so that nearby samples contribute more strongly to calibration.  

Following the same principle as in Section~\ref{sec:COLA-marginal}, the individualized allocation is then determined by solving
\begin{align}\label{eq:alpha_local}
\widehat{\alphabold} (\mX_{n+1}) \in  \arg \min_{\alphabold \in \boldsymbol{\Theta}} | \cap_{k = 1}^K \widehat{C}_{k}^{\mathrm{loc}} (\mX_{n+1};\alpha_k) |,
\end{align}
and the aggregated prediction set is 
\[	\widehat{C}^{\mathrm{loc}}(\mX_{n+1};\alpha) = \cap_{k = 1}^K \widehat{C}_{k}^{\mathrm{loc}} (\mX_{n+1};\widehat{\alpha}_k(\mX_{n+1})).\]
In contrast to Eq.(\ref{eq:alpha_eff}), which optimizes a global average, $\mathrm{COLA\text{-}l}$ focuses solely on minimizing the size of the prediction set at the specific test point $\mX_{n+1}$.

While we focus on localized conformal prediction to construct the individual sets $\{\widehat{C}^{\mathrm{loc}}_k\}_{k=1}^K$, the proposed framework is applicable beyond this setting. In principle, it can incorporate any collection of conformal prediction sets that satisfy (asymptotic) conditional coverage, including alternatives such as \citet{chernozhukov2021distributional,gibbs2025conformal}.

\begin{remark}
The stepwise optimization strategy in Algorithm~\ref{alg:stepwise} also applies to problem~(\ref{eq:alpha_local}). Although $\widehat{C}^{\mathrm{loc}}_k(\mX_{n+1};\alpha_k)$ changes at the jump points of the weighted empirical distribution $\sum_{i\in [n]} w_i \delta_{S^k_i}$ rather than the uniform grid points $\Gcal(n)$, the discretized feasible set $\boldsymbol{\Theta}_{\mathrm{grid}}(\alpha, n)$ is sufficiently fine such that the algorithm remains effective for problem~(\ref{eq:alpha_local}).
\end{remark}

\subsection{Asymptotic Guarantee of COLA-l}
Before proceeding, we introduce the population-level conditional oracle prediction sets. 
For each score function $S_k$, define 
\begin{align*}
	C_{k}^{\mathrm{loc}}(\mX_{n+1};\alpha_k) = \{y: S_k(\mX_{n+1},y) \leq q_{k,\alpha_k}(\mX_{n+1}) \},
\end{align*}
where $q_{k,\alpha}(\mX_{n+1}) = \inf\{t \in \real: \P(S_k(\mX_{n+1},Y_{n+1}) \leq t | \mX_{n+1} ) \geq 1 - \alpha \} $ is the $(1-\alpha)$ quantile of the score $S_k(\mX_{n+1},Y_{n+1})$ conditional on $\mX_{n+1}$.

To analyze the conditional coverage and the size efficiency of $\mathrm{COLA\text{-}l}$, we consider the feature space as $\Xcal \subset \mathbb{R}^d$ and adopt the Laplace kernel $H(\mX_1,\mX_2) = \exp(-{\norm{\mX_1 - \mX_2}}/{h_n} )$ with bandwidth parameter $h_n$. 
Similar results apply to other common kernels. 
The following conditions are needed.

\begin{assumption}	\label{ass:local_ass}
	There exist constants $\tau$, $\rho$, $l_\alpha$  and $L_q$ such that the following conditions hold.
	\begin{itemize}
		\item [(a).] The conditional cumulative distribution functions $F_{k} (q|\vx)$ of $S_k(\mX_{n+1},Y_{n+1})$ given $\mX_{n+1} = \vx$ are continuous, and satisfy $\sup_{k \in [K]} \sup_{q \in \real} \abs{F_{k}(q|\vx) - F_{k}(q|\vx^\prime)} \leq \tau \norm{\vx - \vx^\prime}$.
		\item [(b).] The density function of the covariate $ \mX_{n+1} $ exists and is lower bounded by $\rho$.
		\item [(c).] The bandwidth $h_n$ satisfies $\lim_n h_n = 0$ and $\lim_n nh_n^d = \infty$.
		\item [(d).] The conditional quantile functions $q_{k,\alpha}(\mX_{n+1})$ are $L_q$-Lipschitz continuous over $[ l_\alpha - \xi , \alpha +\xi ]$ with $\xi = O\big((\log(K n \rho h_n) / (n \rho V_d h_n^d))^{1/2} + \tau h_n \log(h_n^{-d})/\rho \big)$, such that $\abs{q_{k,\alpha_1}(\mX_{n+1}) - q_{k,\alpha_2}(\mX_{n+1})} \leq L_q \abs{\alpha_1 - \alpha_2}$ for any $\alpha_1$, $\alpha_2 \in [ l_\alpha - \xi , \alpha +\xi ]$, where $V_d$ is the volume of the unit ball in $\real^d$.
		\item [(e).] For any $\alphabold = (\alpha_k)_{k \in [K]} \in \boldsymbol{\Theta}$ and $\alpha_k^\prime = \max(\alpha_k, l_\alpha)$, it holds that
		$\cap_{k = 1}^K C_{k}^{\mathrm{loc}}(\vx;\alpha_k)  = \cap_{k = 1}^K C_{k}^{\mathrm{loc}}(\vx;\alpha_k^\prime)$ and
		$\cap_{k = 1}^K \widehat{C}_{k}^{\mathrm{loc}}(\vx;\alpha_k) = \cap_{k = 1}^K \widehat{C}_{k}^{\mathrm{loc}}(\vx;\alpha_k^\prime).$
	\end{itemize}
\end{assumption}

Assumptions~\ref{ass:local_ass}(a)-(d) ensure uniform convergence of the conditional quantile estimators, as discussed in \cite{guan2022localized}. Assumption~\ref{ass:local_ass}(e) serves as a localized counterpart of Assumption~\ref{ass:marginal_ass}(c). We now state the main results.
\begin{theorem}\label{thm:valid_conditional}
	\begin{itemize}
		\item [(1)] Under Assumptions~\ref{ass:local_ass}(a)-(c), it follows that
		\[\P(Y_{n+1}\in\widehat{C}^{\mathrm{loc}}(\mX_{n+1}; \alpha) | \mX_{n+1} ) \geq 1-\alpha - K ~ O \left( \sqrt{\frac{\log(K n \rho h_n)}{n \rho h_n^d}} +  \frac{\tau h_n \log(h_n^{-d})}{\rho} \right).\]
		\item [(2)] Suppose Assumptions~\ref{ass:common_ass} and \ref{ass:local_ass} hold. If the conditional density $f(y|\vx)$ of $Y_{n+1}|\mX_{n+1} = \vx$ exists with $\sup_{\vx \in \Xcal,y \in \Ycal} f(y|\vx) \leq M_f$, then it follows that
		\[\P(Y_{n+1}\in\widehat{C}^{\mathrm{loc}}(\mX_{n+1}; \alpha) | \mX_{n+1} ) \geq 1-\alpha - m M_f L_q L_s ~  O\left( \sqrt{\frac{\log(K n \rho h_n)}{n \rho h_n^d}} + \frac{\tau h_n \log(h_n^{-d})}{\rho} \right).\]
	\end{itemize}
\end{theorem}
The dependence on $K$ in Theorem~\ref{thm:valid_conditional}-(1) appears outside the square root. This linear dependence on $K$ stems from the accumulation of estimation errors across the $K$ conditional quantiles. Under stronger conditions, the rate in Theorem~\ref{thm:valid_conditional}-(2) improves upon that in Theorem~\ref{thm:valid_conditional}-(1), reducing the dependence on $K$ to $\log(K)$.

Next, we turn to the size efficiency of COLA-l. 
\begin{theorem}\label{thm:optimal_conditional}
	Under Assumptions~\ref{ass:common_ass} and \ref{ass:local_ass}, conditional on $\mX_{n+1}$, it follows that 
	\begin{align*}
		\abs{\widehat{C}^{\mathrm{loc}}(\mX_{n+1}; \alpha)} - \Lcal^*(\mX_{n+1})
		= m L_q L_s ~ O\left( \sqrt{\frac{\log(K n \rho h_n)}{n \rho h_n^d}} + \frac{\tau h_n \log(h_n^{-d})}{\rho} \right),
	\end{align*}
	where $\Lcal^*(\mX_{n+1})=\min_{\alphabold \in \boldsymbol{\Theta}}| \cap_{k = 1}^K C_{k}^{\mathrm{loc}}(\mX_{n+1};\alpha_k)|$.
\end{theorem}

Assuming $K = O(n^C)$ for constant $C>0$ and choosing $h_n \asymp n^{-{1}/{(d + 2)}}$, COLA-l achieves asymptotic conditional coverage and locally oracle size $\Lcal^*(\mX_{n+1})$ at the rate $n^{-{1}/{(d + 2)}} \log(n)$. Table~\ref{tab:cola_comparison} in the Supplementary Material summarizes the characteristics of the four COLA variants: $\mathrm{COLA\text{-}e}$, $\mathrm{COLA\text{-}s}$, $\mathrm{COLA\text{-}f}$, and $\mathrm{COLA\text{-}l}$.

\section{Numerical Result}\label{sec:numerical_result}
In this section, we examine the finite-sample performance of COLA, focusing on COLA-e and COLA-s for average-case allocation and COLA-l for individualized allocation. The goals are to empirically (i) verify the nominal coverage guarantees, and (ii) assess the efficiency. Throughout, we fix the nominal miscoverage level $\alpha=0.1$. The stepwise optimization procedure in Algorithm~\ref{alg:stepwise} is applied with $k_{\mathrm{max}} = 4$. Those methods that require data splitting use an equal split of the hold-out set $\Dcal$, i.e., $n_{\mathrm{tu}}=n_{\mathrm{cal}}=n/2$.

\subsection{Marginal-level Performance}\label{sec:numerical_marginal}
A full comparison among COLA-e, COLA-s, and COLA-f (including coverage, set size, and runtime) is presented in Section~\ref{subsec:comparison_s_e_f} of the Supplementary Material. Due to the high computational cost of COLA-f and its empirical similarity to COLA-e and COLA-s, the main numerical studies below concentrate on COLA-e and COLA-s. In this section, we examine three representative cases where the data are generated from $Y=\mu(\mX) + \varepsilon$ with $\mX = (X_j)_{j \in [d]}$:
\begin{itemize}
\item Case 1 (different learning algorithms): 
$ \mu(\mX) = \indicator\left\{X_1+X_2+X_3 > 0\right\}$, $\mX \sim \normal(0,\Sigma)$ with $d=5$ and $\Sigma_{i,j}=(0.5)^{\abs{i-j}}$, and $ \varepsilon \sim \normal(0,0.01).$  $\widehat{\mu}$ is trained using four distinct algorithms including ordinary least squares (OLS), random forest (RF), neural network, and multivariate adaptive regression splines (MARS), producing $K = 4$ residual scores.
\item Case 2 (different nonconformity measures): Same $\mu(\mX)$ and $\mX$ as Case 1 with $\varepsilon \sim \normal(0,(0.03 X_1)^2).$
Three types of scores are considered: residual scores using $\widehat{\mu}$ from MARS; rescaled residual scores using $\widehat{\mu}$ and $\widehat{\sigma}$ from RF fitted separately; and conformalized quantile regression scores using $\widehat{\tau}_{\alpha/2}$ and $\widehat{\tau}_{1-\alpha/2}$ from quantile regression forests, yielding $K = 3$.
\item Case 3 (different feature subsets):
$\mu(\mX) = \mX^\top \betabold $, where $ \beta_j = \indicator\{j \bmod 20 = 0\}$ for $j\in[d]$ with $d=100$.  $\mX\sim \mathcal{N}(0,I_{d})$ and  $\varepsilon \sim \normal(0,1).$ 
Four ridge regression submodels are fitted, each using a randomly selected subset of $20$ features and a regularization parameter of $0.1$, producing $K = 4$ residual scores.
\end{itemize}

We benchmark COLA-e and COLA-s against five alternatives: efficiency-first (EFCP) and validity-first (VFCP) conformal prediction algorithms from \cite{yang2024selection}; the Majority Vote algorithm from \cite{gasparin2024merging}; the Synthetic, Aggregation, and Test inversion (SAT) approach from \cite{qin2024sat}; and a naive baseline (Random) that selects one score uniformly at random. In Majority Vote, each of the $K$ individual sets is constructed at level $1-\alpha/2$, meeting the nominal coverage guarantee. SAT aggregates conformal p-values using the Cauchy combination test \citep{liu2020cauchy} as recommended in \cite{qin2024sat}. 

For each experimental setting, we repeat $200$ independent trials, each involving an independent draw of a training set of size $150$, a test set of size $40$, and a hold-out set $\mathcal{D}$ of size $n$. The test set is used to evaluate empirical average coverage and set size, while the hold-out set $\mathcal{D}$ is used to construct prediction sets. Figure~\ref{fig:lineplot_n} reports the results as $n$ varies from $100$ to $600$, showing average coverage (top row) and average set size (bottom row) scaled by that of COLA-e. 

\begin{figure}[t!]
\centering
\begin{minipage}{0.8\textwidth}
	\centering
	\includegraphics[trim=4.9cm 14.4cm 4.9cm 0cm, clip, height=0.4cm]{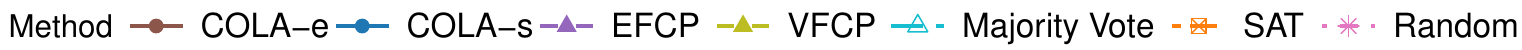}
\end{minipage}
\vspace{0.1em}
\begin{minipage}{0.32\textwidth}
	\includegraphics[trim=1cm 0 7.3cm 0, clip, height=9cm]{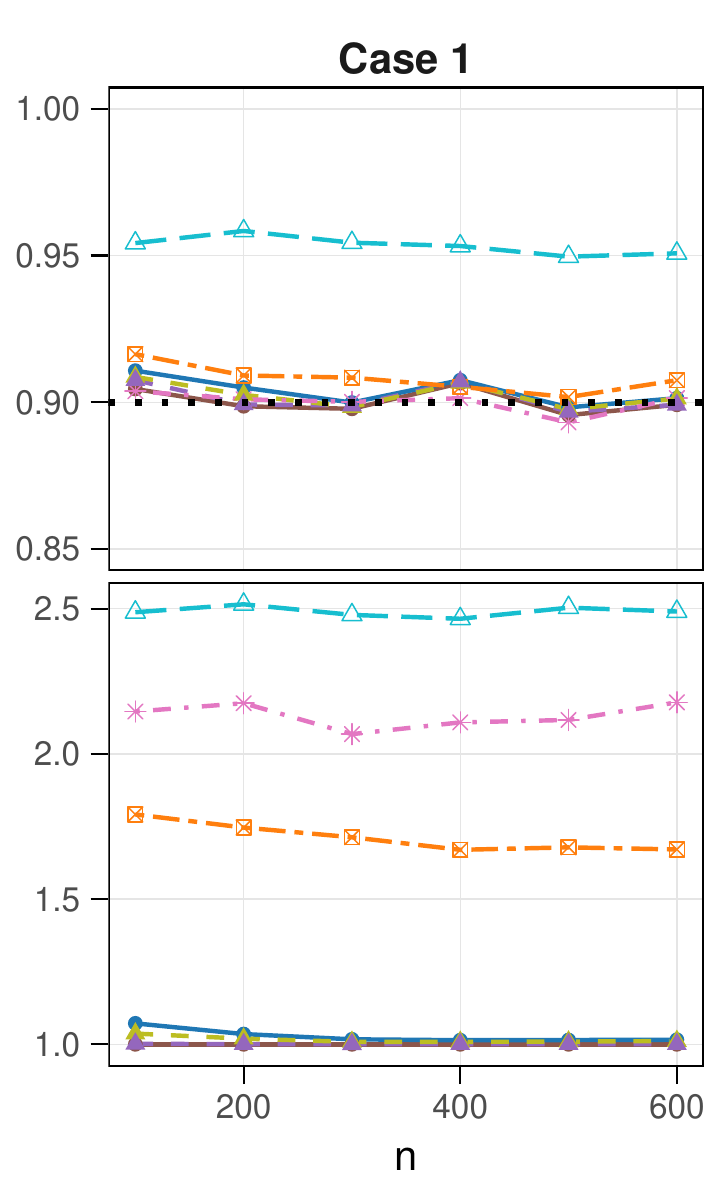}
\end{minipage}%
\hfill
\begin{minipage}{0.32\textwidth}
	\includegraphics[trim=1cm 0 7.3cm 0, clip, height=9cm]{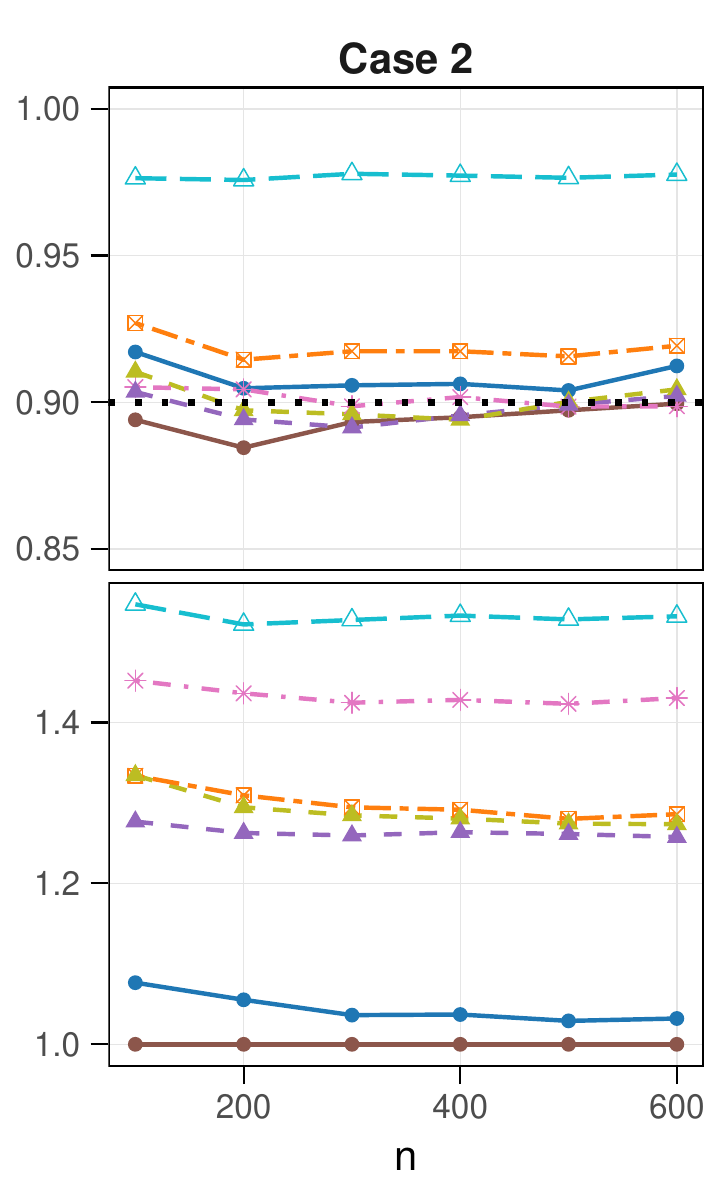}
\end{minipage}%
\hfill
\begin{minipage}{0.34\textwidth}
	\includegraphics[trim=1cm 0 5.5cm 0, clip, height=9cm]{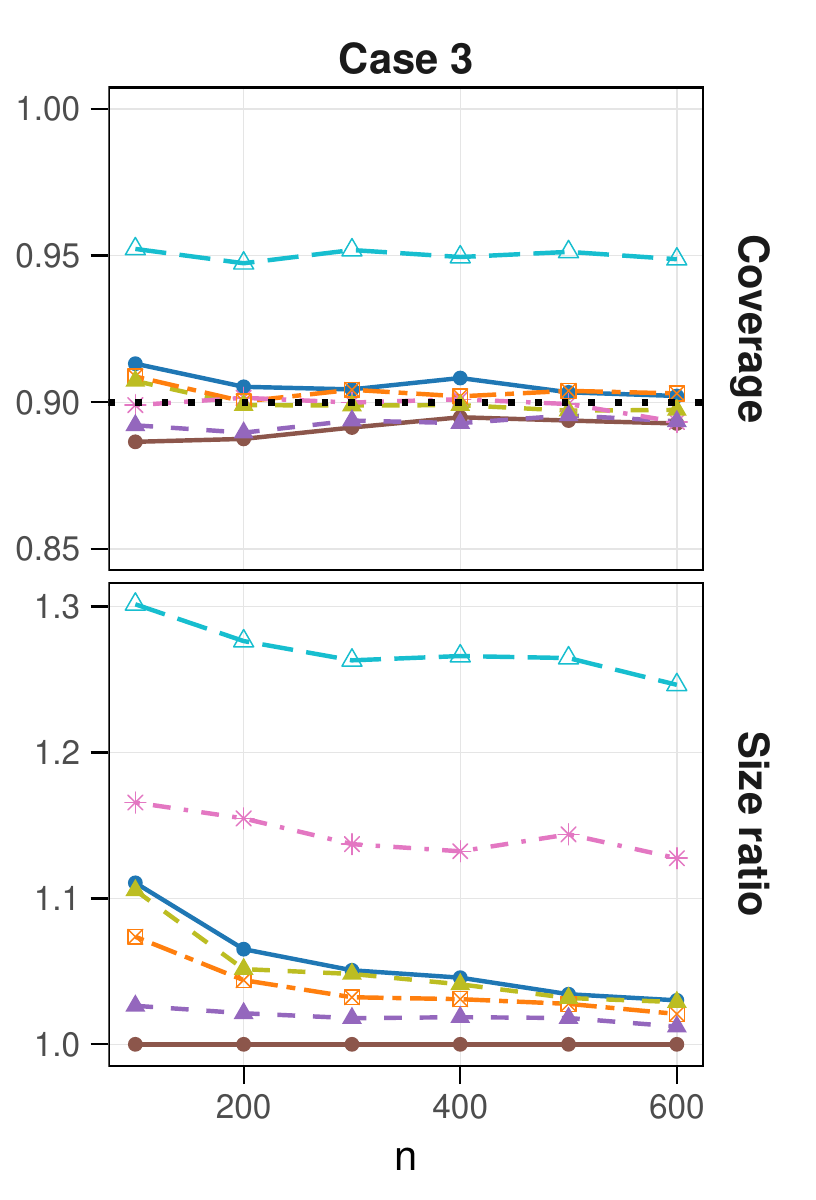}
\end{minipage}
\caption{Performance under Cases 1--3 with varying $n$. Top row: coverage; bottom row: prediction set size ratio relative to COLA-e. The black dotted lines mark the nominal coverage level $1-\alpha=0.9$.}
\label{fig:lineplot_n}
\end{figure}

Across all three cases, COLA-e and COLA-s perform competitively. In Case 1, their results are broadly comparable to EFCP and VFCP (see Section~\ref{sec:numerical_case1} in the Supplementary Material for detailed comparisons). In Case 2, both COLA-e and COLA-s produce substantially smaller sets than EFCP and VFCP, highlighting the advantage of allocation over single-score selection. Meanwhile, COLA-e exhibits slightly undercoverage in both Case 2 and Case 3, similar to EFCP, but this effect diminishes as $n$ increases, consistent with the asymptotic validity of COLA-e established in Theorem~\ref{thm:valid_cola_e}. In contrast, set-level combination methods such as Majority Vote and SAT generally produce longer prediction sets.
This underscores the benefit of level-wise allocation: by directly optimizing confidence-level allocation across candidate scores, COLA-e and COLA-s can achieve a more favorable tradeoff between coverage and efficiency. 

We next examine the effect of varying $K$ under Case 3. Keeping all configurations the same as before, we vary $K$ from $2$ to $128$ while fixing the hold-out size at $n=300$. The results are reported in Figure~\ref{fig:lineplot_K}. COLA-e and EFCP show some undercoverage that becomes more pronounced as 
$K$ increases, consistent with theoretical results. In terms of set size, COLA-e and COLA-s consistently outperform EFCP and VFCP, whereas Majority Vote and SAT produce increasingly larger prediction sets as $K$ grows.

\begin{figure}[H]
\centering
\includegraphics[width=1\linewidth]{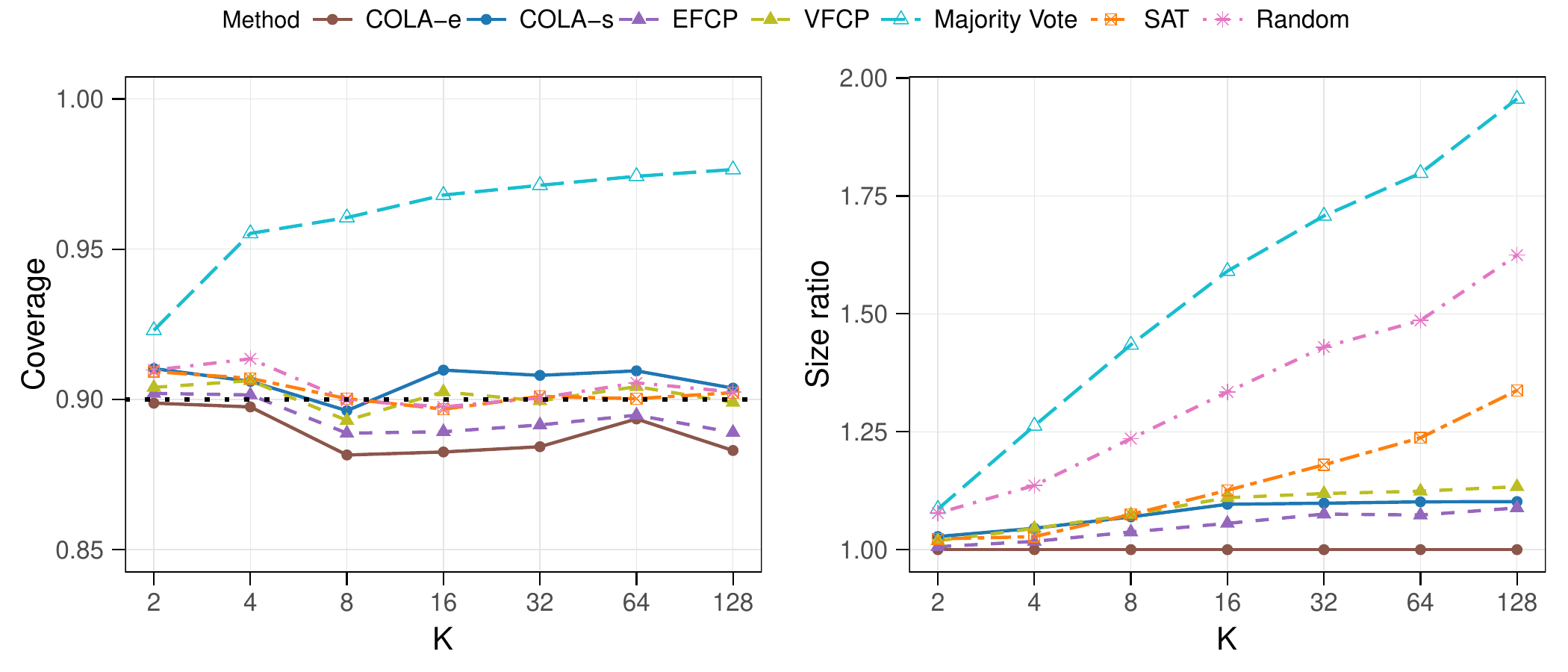}
\caption{Performance under Case 3 with varying $K$. Left: coverage; right: prediction set size ratio relative to COLA-e. The black dotted line marks the nominal coverage level $1-\alpha=0.9$.}
\label{fig:lineplot_K}
\end{figure}

More comprehensive experimental results are provided in Section~\ref{sec:numerical_supp} of the Supplementary Material, including comparisons with score-level combination \citep{luo2025weighted} and model-level combination \citep{yang2024selection}, as well as results under varying miscoverage levels $\alpha$.

\subsection{Individual-level Performance}\label{sec:numerical_individual}

We now turn to the individualized setting and compare COLA-l with COLA-e to highlight the benefits of individualized allocation. We consider $Y = \mu(X) + \varepsilon$ with 
$\mu(X) = X \indicator\{-1 < X < 1\} +  \indicator\{X \geq 1\} - \indicator\{X \leq - 1\}$,
$X \sim \mathrm{Unif}[-2,2]$ and $\varepsilon \sim \normal(0,(0.25 + 0.25\abs{X})^2)$. We use $K=2$ residual scores obtained from OLS and regression trees. Both COLA-l and COLA-e are equipped with the localized conformal prediction \citep{guan2022localized} to construct the conformal prediction sets with a Laplace kernel $H(x,x') = \exp(-\abs{x - x'}/h_n)$ and bandwidth $h_n = c n^{-1/(d+2)}$, where the constant $c$ is chosen to ensure an effective sample size of $200$ \citep{hore2025conformal}. The regression models are trained on $1000$ samples; the hold-out data $\mathcal{D}$ has $n=2000$ or $n=5000$ observations. Conditional coverage and set size are evaluated at 21 equally spaced $X\in[-1,1]$ using $2000$ new samples per location.

Figure~\ref{fig:individual_plot} shows that both COLA-l and COLA-e maintain the conditional coverage guarantee, with stability improving as $n$ grows. Importantly, COLA-l produces smaller prediction sets across most locations by adaptively allocating $\alpha$ according to the covariate $X$, whereas COLA-e uses one allocation across all $X$. This demonstrates the advantage of individualized allocation: by tailoring the allocation to the covariate space, COLA-l delivers finer efficiency.
\begin{figure}[H]
\centering
\includegraphics[width=1\linewidth]{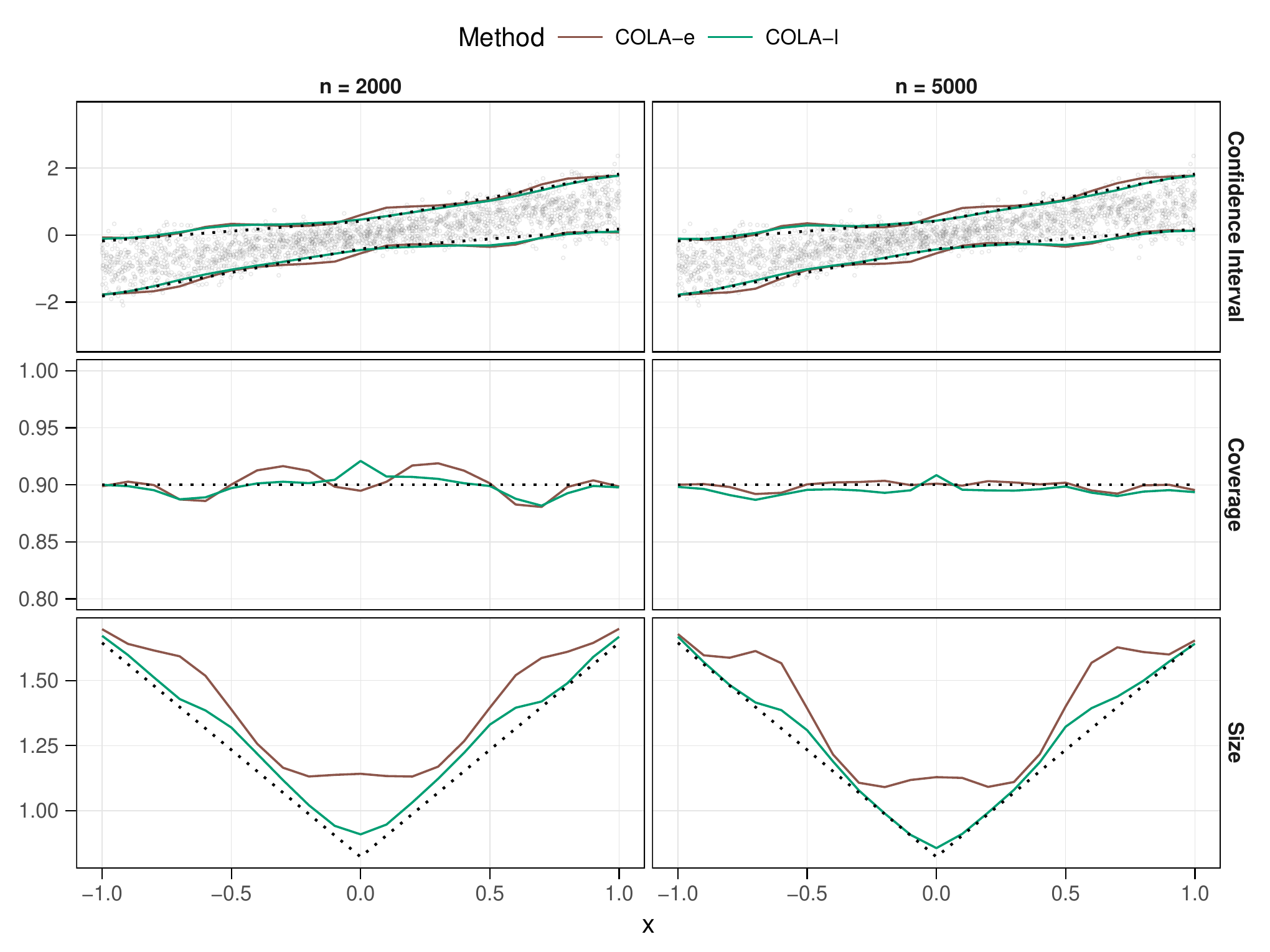}
\caption{Performance of COLA-e and COLA-l across locations. Top row: prediction set (black dotted line, oracle conditionally valid set; gray points, data samples). Middle row: coverage (black dotted line, nominal level $1-\alpha=0.9$). Bottom row: prediction set size (black dotted line, size of the oracle conditionally valid set).}
\label{fig:individual_plot}
\end{figure}

\subsection{Real Data Example}
We evaluate the competing methods on $3$ regression datasets from the UCI repository: \textit{Blog Feedback, Concrete Strength, Superconductivity} (see Appendix G of \cite{gupta2022nested} for dataset details). For all datasets, we construct the prediction sets using the same three nonconformity measures and corresponding prediction rules as in Case 2. For each dataset, 1000 observations are randomly subsampled without replacement into a training set of size 384, a test set of size 232, and a hold-out set $\Dcal$ of size $n=384$, serving the same roles as in earlier experiments. The above procedure is independently repeated $100$ times, and the reported metrics are the average coverage and average set size. 

As shown in Table~\ref{tab:uci_reduced}, our method consistently achieves near-nominal coverage while producing substantially smaller prediction sets compared to other baselines. This demonstrates both the reliability and efficiency of our approach. Similar results on three additional datasets are provided in Table~\ref{tab:uci_remaining} of the Supplementary Material. 

\begin{table}[H]
\centering
\caption{Coverage and prediction set size on three UCI datasets ($\alpha = 0.1$).}
\label{tab:uci_reduced}
\begin{tabular}{
		l
		>{\centering\arraybackslash}p{1.5cm}
		>{\centering\arraybackslash}p{1.5cm}
		>{\centering\arraybackslash}p{1.5cm}
		|
		>{\centering\arraybackslash}p{1.5cm}
		>{\centering\arraybackslash}p{1.5cm}
		>{\centering\arraybackslash}p{1.5cm}
	}
	\toprule
	\multirow{2}{*}{\textbf{Method}} 
	& \multicolumn{3}{c|}{\textbf{Coverage (\%)}} 
	& \multicolumn{3}{c}{\textbf{Size}} \\
	& \textit{Blog} & \textit{Concrete} & \textit{Conduct} 
	& \textit{Blog} & \textit{Concrete} & \textit{Conduct} \\
	\midrule
	$\mathrm{COLA\text{-}e}$ & 91.97 & 90.24 & 90.34   & 13.36 & 19.64 & 43.19 \\
	$\mathrm{COLA\text{-}s}$       & 92.83 & 91.28 & 91.55   & 17.57 & 20.78 & 45.93 \\
	EFCP    & 92.43 & 89.90 & 90.06   & 22.39 & 20.85 & 46.37 \\
	VFCP         & 92.85 & 90.29 & 90.13   & 23.01 & 21.21 & 47.07 \\
	Majority Vote  & 97.53 & 97.76 & 96.95   & 48.10 & 29.91 & 60.21 \\
	SAT & 60.51 & 92.69 & 92.14   & 25.76 & 21.78 & 45.94 \\
	Random     & 92.22 & 90.03 & 90.20   & 45.79 & 26.23 & 55.85 \\
	\bottomrule
\end{tabular}
\end{table}

\section{Conclusion and Future Work}\label{sec:conclusion}
This work introduces COLA, a general framework for aggregating multiple conformal prediction sets by optimally allocating the miscoverage budget $\alpha$ across candidate sets and intersecting them. COLA achieves substantial efficiency gains with rigorous coverage guarantees and attains asymptotically optimal allocation under mild conditions. In the average-case setting, its three variants including COLA-e, COLA-s, and COLA-f, balance efficiency, validity, and computational cost (see Table~\ref{tab:cola_comparison}). We further propose COLA-l, which adapts allocation to each test point while ensuring conditional coverage and size efficiency. A stepwise optimization algorithm is proposed to enhance scalability and practical implementation. Extensive experiments show that COLA yields smaller prediction sets than existing baselines, highlighting its flexibility and broad applicability.

The findings of this work open several directions for future research. First, a theoretical analysis of COLA-f’s efficiency, in the spirit of  \cite{liang2024conformal}, is needed to complement its strong empirical performance. Second, more computationally effective optimization procedures are required to find global solutions in large-scale settings reliably. 
Finally, beyond size efficiency, other notions of optimality deserve attention, such as allocating the budget $\alpha$ to maximize the utility of prediction-based decisions under the coverage constraint, both on average and at the individual level.

\bibliographystyle{asa}
\bibliography{ref}

\appendix
\newpage

\begin{center}
    \linespread{2}\selectfont
    {\Large \bfseries Supplementary Material for ``Aggregating Conformal Prediction Sets via $\alpha$-Allocation"}
\end{center}
\vspace{30pt}
The Supplementary Material contains all technical proofs, additional algorithmic and theoretical discussions, including the intuition behind our method, and details of the stepwise algorithm, as well as extended numerical studies and implementation details. 
\section{Technical Proofs}
\subsection{Proof of Theorem \texorpdfstring{\ref{thm:valid_cola_e}}{1} }\label{proof:optimal_efcp}
\begin{proof}[Proof of Theorem~\ref{thm:valid_cola_e}]
	For notational convenience, define
	\[
	W = \sup_{\vt} \Bigg| \frac{1}{n} \sum_{i\in[n]} \prod_{k=1}^K \indicator\lrl{S_{k,i} \leq t_k}
	- \E\lrm{ \prod_{k=1}^K \indicator\lrl{S_k(\mX_{n+1}, Y_{n+1}) \leq t_k} } \Bigg|.
	\]
	We have
	\begin{align*}
		&\P(Y_{n+1} \in \widehat{C}^{\mathrm{e}}(\mX_{n+1};\alpha) \mid \Dcal) \\
		&= \P\lrm{ \cap_{k=1}^K \lrl{ S_k(\mX_{n+1}, Y_{n+1}) \leq Q_{\widehat{\alpha}^{\mathrm{e}}_k}(\{S_{k,i}\}_{i\in[n]} \cup \{\infty\}) } \mid \Dcal } \\
		&= \E\lrm{ \prod_{k=1}^K \indicator\lrl{ S_k(\mX_{n+1}, Y_{n+1}) \leq Q_{\widehat{\alpha}^{\mathrm{e}}_k}(\{S_{k,i}\}_{i\in[n]} \cup \{\infty\}) } \mid \Dcal } \\
		&\geq \frac{1}{n} \sum_{i\in[n]} \prod_{k=1}^K \indicator\lrl{ S_{k,i} \leq Q_{\widehat{\alpha}^{\mathrm{e}}_k}(\{S_{k,i}\}_{i\in[n]} \cup \{\infty\}) } - W \\
		&= \frac{1}{n} \sum_{i\in[n]} \indicator\lrl{ Y_i \in \widehat{C}^{\mathrm{e}}(\mX_{i};\alpha) } - W.
	\end{align*}
	Subtracting from $1$ gives
	\begin{align*}
		\P\lrl{ Y_{n+1} \notin \widehat{C}^{\mathrm{e}}(\mX_{n+1};\alpha) \mid \Dcal }
		&\leq \frac{1}{n} \sum_{i\in[n]} \indicator\lrl{ Y_{i} \notin \cap_{k=1}^K \widehat{C}_k(\mX_{i};\widehat{\alpha}_k^{\mathrm{e}}) } + W \\
		&\leq \frac{1}{n} \sum_{i\in[n]} \sum_{k\in[K]} \indicator\lrl{ Y_{i} \notin \widehat{C}_k(\mX_{i};\widehat{\alpha}_k^{\mathrm{e}}) } + W \\
		&= \sum_{k\in[K]} \frac{1}{n} \sum_{i\in[n]} \indicator\lrl{ S_{k,i} > Q_{\widehat{\alpha}^{\mathrm{e}}_k}(\{S_{k,i}\}_{i\in[n]} \cup \{\infty\}) } + W \\
		&\leq \sum_{k\in[K]} \frac{1}{n} \sum_{i\in[n]} \indicator\lrl{ S_{k,i} > Q_{\widehat{\alpha}^{\mathrm{e}}_k}(\{S_{k,i}\}_{i\in[n]}) } + W \\
		&\leq \sum_{k\in[K]} \widehat{\alpha}_k^{\mathrm{e}} + W = \alpha + W.
	\end{align*}
	Taking expectations, we obtain
	\begin{equation}\label{eq:original_cov_bound}
		\P\lrl{ Y_{n+1} \notin \widehat{C}^{\mathrm{e}}(\mX_{n+1};\alpha) } \leq \alpha + \E[W].
	\end{equation}
	
	For any  $t \geq 0$, we have
	\begin{equation}
		\begin{aligned}\label{eq:bound_on_W}
			\E[W] & = \E[W \indicator\{W > t\}] + \E[W \indicator\{W \leq t\}]\\
			& \leq \Pr(W>t) + t\\
			& \leq K(n+1) \exp(-2nt^2) + t,
		\end{aligned}
	\end{equation}
	where the first inequality uses the bound $W \leq 1$ and the second follows from Lemma 4.1 in \cite{naaman2021tight}. Choosing $t = \sqrt{\log(nK)/n}$ gives
	\[
	\E[W] = O\lrs{ \sqrt{\frac{\log(nK)}{n}} }.
	\]
	We conclude the proof by substituting the above result into \eqref{eq:original_cov_bound}:
	\[
	\P\lrl{ Y_{n+1} \in \widehat{C}^{\mathrm{e}}(\mX_{n+1};\alpha) }
	\geq 1 - \alpha - O\lrs{ \sqrt{\frac{\log (Kn)}{n}} }.
	\] 
\end{proof}

\subsection{Proof of Theorem \texorpdfstring{\ref{thm:asymptotic_opt_cola_e}}{2} } \label{proof:eff_optimal}
The following two lemmas are required for the proof of Theorem~\ref{thm:asymptotic_opt_cola_e}, with their proofs given at the end of this section.
\begin{lemma}\label{lemma:quantile}    
	Let $x_1,\dots,x_n$ be i.i.d.\ samples from a distribution with continuous CDF and quantile function $q(\alpha)$. 
	Assume $q(\alpha)$ is $L_q$-Lipschitz on $[a,b] \subset [0,1]$, where $a,b$ are constants. 
	Then, for any $\delta \in (0,1)$, with probability at least $1-\delta$,
	\begin{align*}
		\sup_{\alpha \in \Big[a + \sqrt{\frac{\log(2/\delta)}{2n}}, \frac{n}{n+1}\Big(b - \sqrt{\frac{\log(2/\delta)}{2n}} \Big)  \Big]} \abs{\widehat{q}^+(\alpha) - q(\alpha)} & \leq L_q  \sqrt{\frac{\log(2/\delta)}{2n}} + \frac{L_q}{n},
	\end{align*}
	where $\widehat{q}^+(\alpha)$ is the empirical quantile function based on the augmented sample $\{x_1,\dots,x_n,\infty\}$.
\end{lemma}

\begin{lemma}\label{lemma:sym_dif}
	(i) Let $A_k = [a_{k,l},a_{k,r}]$ and $B_k = [b_{k,l},b_{k,r}]$ be intervals for $k \in [K]$. 
	If $\abs{A_k \triangle B_k} \le c$ for all $k$, then
	\[
	\abs{\cap_{k=1}^K A_k \ \triangle \ \cap_{k=1}^K B_k} \le 2c.
	\]
	
	(ii) Let $A_k = \bigcup_{i=1}^{m_{a,k}} [a_{k,i,l},a_{k,i,r}]$ and 
	$B_k = \bigcup_{i=1}^{m_{b,k}} [b_{k,i,l},b_{k,i,r}]$ be finite unions of intervals for $k \in [K]$. 
	If $\abs{A_k \triangle B_k} \le c$ for all $k$, then
	\[
	\abs{\cap_{k=1}^K A_k \ \triangle \ \cap_{k=1}^K B_k} \le 4m c,
	\]
	where $m = \max_{k \in [K]} \{m_{a,k}, m_{b,k}\}$.
\end{lemma}
\begin{proof}[Proof of Theorem~\ref{thm:asymptotic_opt_cola_e}]
	We first bound $\Lcal(\widehat{\alphabold}^\mathrm{e}) - \Lcal(\alphabold^*)$. Observe that
	\[
	0 \le \Lcal(\widehat{\alphabold}^\mathrm{e}) - \Lcal(\alphabold^*) 
	\le \Lcal(\widehat{\alphabold}^\mathrm{e}) - \Lcal_n(\widehat{\alphabold}^\mathrm{e}) + \Lcal_n(\alphabold^*) - \Lcal(\alphabold^*) 
	\le 2 \sup_{\alphabold} \abs{\Lcal_n(\alphabold) - \Lcal(\alphabold)},
	\]
	where the first inequality uses the optimality of $\alphabold^*$ and the second uses that of $\widehat{\alphabold}^\mathrm{e}$. Decompose
	\begin{align*}
		\sup_{\alphabold} \abs{\Lcal_n(\alphabold) - \Lcal(\alphabold)}& =\sup_{\alphabold} \Bigabs{ \frac{1}{n} \sum_{i \in [n]} | \cap_{k = 1}^K \widehat{C}_k (\mX_{i};\alpha_k) |-\E_{\mX_{n+1}}\big[|\cap_{k=1}^KC_k(\mX_{n+1};\alpha_k)|\big]}\\
		&\le \underbrace{\sup_{\alphabold,\vx}\Bigabs{\cap_{k=1}^K \widehat{C}_k(\vx;\alpha_k) \triangle \cap_{k=1}^K C_k(\vx;\alpha_k)}}_{\Pi_1} \\
		&\quad + \underbrace{\sup_{\alphabold} \Bigabs{\frac{1}{n}\sum_i \abs{\cap_{k=1}^K C_k(\mX_i;\alpha_k)} - \E|\cap_{k=1}^K C_k(\mX_{n+1};\alpha_k)|}}_{\Pi_2}.
	\end{align*}
	
	We begin by bounding the term $\Pi_1$. For brevity, we denote $\widehat{q}^+_{k,\alpha} = Q_\alpha(\{S_{k,i}\}_{i \in [n]} \cup \{\infty\})$. If
	\[
	[1-\alpha, 1-l_\alpha] \subset \Big[1-\alpha-\xi + \sqrt{\frac{\log(2/\delta)}{2n}}, \frac{n}{n+1}\big(1-l_\alpha+\xi - \sqrt{\frac{\log(2/\delta)}{2n}}\big) \Big],
	\]
	Lemma~\ref{lemma:quantile} gives
	\[
	\sup_{\beta \in [l_\alpha,\alpha]} \abs{q_{k,\beta} - \widehat{q}^+_{k,\beta}} \le L_q \sqrt{\frac{\log(2/\delta)}{2n}} + \frac{L_q}{n},
	\]
	with probability at least $1-\delta$. A union bound over $k \in [K]$ yields
	\[
	\sup_{k \in [K]} \sup_{\beta \in [l_\alpha,\alpha]} \abs{q_{k,\beta} - \widehat{q}^+_{k,\beta}} \le L_q \sqrt{\frac{\log(2K/\delta)}{2n}} + \frac{L_q}{n},
	\]
	provided $\xi \ge \sqrt{\log(2K/\delta)/(2n)} + (2-l_\alpha)/n$. Using Assumption~\ref{ass:common_ass}(b) and the definitions of $C_k$ and $\widehat{C}_k$, we have
	\begin{equation}\label{eq:individual_error_simplified}
		\sup_{k \in [K]} \sup_{\beta \in [l_\alpha,\alpha]} \sup_{\vx} \abs{C_k(\vx;\beta) \triangle \widehat{C}_k(\vx;\beta)} 
		\le L_s \sup_{k,\beta} \abs{q_{k,\beta} - \widehat{q}^+_{k,\beta}} 
		\le L_q L_s \sqrt{\frac{\log(2K/\delta)}{2n}} + \frac{L_q L_s}{n}.
	\end{equation}
	For any $\alphabold \in \boldsymbol{\Theta}$, define $\alpha_k' = \max(\alpha_k, l_\alpha)$. Then, by Assumption~\ref{ass:marginal_ass}(c),
	\[
	\cap_{k=1}^K C_k(\vx;\alpha_k) \triangle \cap_{k=1}^K \widehat{C}_k(\vx;\alpha_k)
	= \cap_{k=1}^K C_k(\vx;\alpha_k') \triangle \cap_{k=1}^K \widehat{C}_k(\vx;\alpha_k').
	\]
	Applying Lemma~\ref{lemma:sym_dif} and \refeq{eq:individual_error_simplified} then gives
	\begin{equation}\label{eq:step_1_2_link}
		\Pi_1 = \sup_{\alphabold,\vx} \abs{\cap_{k=1}^K C_k(\vx;\alpha_k) \triangle \cap_{k=1}^K \widehat{C}_k(\vx;\alpha_k)}
		= O\Big(m L_q L_s \sqrt{\frac{\log(2K/\delta)}{2n}}\Big),
	\end{equation}
	with probability at least $1-\delta$.
	
	Next, we bound $\Pi_2$. Define $g_{\alphabold}(\vx) = \abs{\cap_{k=1}^K C_k(\vx;\alpha_k)}$ and $ \Gcal = \{g_{\alphabold} : \alphabold \in \boldsymbol{\Theta}\}$. For any $\alphabold_j = (\alpha_{j,k})_{k \in [K]} \in \boldsymbol{\Theta}$ ($j=1,2$), define $\alpha_{j,k}' = \max(\alpha_{j,k}, l_\alpha)$. Then $g_{\alphabold}(\vx)$ is Lipschitz in $\alphabold$:
	\begin{align*}
		\abs{g_{\alphabold_1}(\vx) - g_{\alphabold_2}(\vx)} 
		&= \Bigabs{\abs{\cap_{k=1}^K C_k(\vx;\alpha_{1,k}')} - \abs{\cap_{k=1}^K C_k(\vx;\alpha_{2,k}')}} \\
		&\le \Bigabs{\cap_{k=1}^K C_k(\vx;\alpha_{1,k}') \triangle \cap_{k=1}^K C_k(\vx;\alpha_{2,k}')} \\
		&\le \sum_{k\in[K]} \abs{C_k(\vx;\alpha_{1,k}') \triangle C_k(\vx;\alpha_{2,k}')} \\
		&\le L_s \sum_{k\in[K]} \abs{q_{k,\alpha_{1,k}'} - q_{k,\alpha_{2,k}'}} \le L_s L_q \norm{\alphabold_1 - \alphabold_2}_1,
	\end{align*}
	where the equality follows from Assumption~\ref{ass:marginal_ass}(c), the inequalities follow from Assumptions~\ref{ass:common_ass}(b), \ref{ass:marginal_ass}(b).  
	
	Construct an $\epsilon$-net $\Ncal_\epsilon$ of $\boldsymbol{\Theta}$ under $\ell_1$-norm with $\epsilon = \frac{M}{L_q L_s} \sqrt{\frac{K \log(n)}{2n}}$.
	Then $\abs{\Ncal_\epsilon} = O(1/\epsilon^K)$. For fixed $\alphabold \in \Ncal_\epsilon$, since $\abs{g_{\alphabold}(\vx) - g_{\alphabold}(\vx')} \le M$, McDiarmid's inequality gives
	\[
	\P\Big(\Bigabs{\frac{1}{n} \sum_i g_{\alphabold}(\mX_i) - \E[g_{\alphabold}(\mX)]} > u \Big) \le \exp\Big(-\frac{2 n u^2}{M^2}\Big).
	\] 
	A union bound over $\alphabold \in \Ncal_\epsilon$ yields
	\[
	\P\Big(\sup_{\alphabold \in \Ncal_\epsilon} \Bigabs{\frac{1}{n} \sum_i g_{\alphabold}(\mX_i) - \E[g_{\alphabold}(\mX)]} > u \Big) \le O\Big(\frac{1}{\epsilon^K} \exp\Big(-\frac{2 n u^2}{M^2}\Big)\Big),
	\]
	which implies that, with probability at least $1-\delta$,
	\[
	\sup_{\alphabold \in \Ncal_\epsilon} \Bigabs{\frac{1}{n} \sum_i g_{\alphabold}(\mX_i) - \E[g_{\alphabold}(\mX)]} = O\Big(M \sqrt{\frac{K \log(n) + \log(2/\delta)}{2n}}\Big).
	\]
	For any $\betabold \in \boldsymbol{\Theta}$, choose $\alphabold \in \Ncal_\epsilon$ with $\norm{\betabold - \alphabold}_1 \le \epsilon$. Then
	\begin{align*}
		& \Bigabs{\frac{1}{n} \sum_i g_{\betabold}(\mX_i) - \E[g_{\betabold}(\mX)]} - \Bigabs{\frac{1}{n} \sum_i g_{\alphabold}(\mX_i) - \E[g_{\alphabold}(\mX)]} \\
		&\le \frac{1}{n} \sum_i \abs{g_{\betabold}(\mX_i) - g_{\alphabold}(\mX_i)} + \abs{\E[g_{\betabold}(\mX) - g_{\alphabold}(\mX)]} \le 2 L_s L_q \epsilon.
	\end{align*}
	Hence, with probability at least $1-\delta$,
	\[
	\sup_{\alphabold \in \boldsymbol{\Theta}} \Bigabs{\frac{1}{n} \sum_i g_{\alphabold}(\mX_i) - \E[g_{\alphabold}(\mX)]} = O\Big(M \sqrt{\frac{K \log(n) + \log(2/\delta)}{2n}}\Big),
	\]
	which gives
	\[
	\Pi_2 = \sup_{\alphabold} \Bigabs{\frac{1}{n} \sum_{i\in[n]} \abs{\cap_{k=1}^K C_k(\mX_i;\alpha_k)} - \E \abs{\cap_{k=1}^K C_k(\mX_{n+1};\alpha_k)}} = \Op\Big(M \sqrt{\frac{K \log(n)}{n}}\Big).
	\]
	
	Combining $\Pi_1$ and $\Pi_2$ yields
	\begin{align}\label{eq:first_result}
		\Lcal(\widehat{\alphabold}^\mathrm{e}) - \Lcal(\alphabold^*) = \Op\left( m L_q L_s \sqrt{\frac{\log(K)}{n}}  + M \sqrt{\frac{K \log(n)}{n}}\right).
	\end{align}    
	
	We conclude the proof by
	\begin{align*}
		\E_{\mX_{n+1}} [\abs{\widehat{C}^{\mathrm{e}}(\mX_{n+1};\alpha)}] - \Lcal(\alphabold^*) & = \E_{\mX_{n+1}} \Big[ \Big| \cap_{k = 1}^K \widehat{C}_k (\mX_{n+1};\widehat{\alpha}_{k}^\mathrm{e}) \Big|-\Big|\cap_{k=1}^KC_k(\mX_{n+1};\alpha_k^*)\Big| \Big]\\
		& = \E_{\mX_{n+1}} \Big[ \Big| \cap_{k = 1}^K \widehat{C}_k (\mX_{n+1};\widehat{\alpha}_{k}^\mathrm{e}) \Big|-\Big|\cap_{k=1}^KC_k(\mX_{n+1};\widehat{\alpha}_{k}^\mathrm{e})\Big| \Big]\\
		& ~~~~ + \E_{\mX_{n+1}} \Big[ \Big|\cap_{k=1}^KC_k(\mX_{n+1};\widehat{\alpha}_{k}^\mathrm{e})\Big|-\Big|\cap_{k=1}^KC_k(\mX_{n+1};\alpha_k^*)\Big| \Big]\\
		& \leq \Pi_1 + \Lcal(\widehat{\alphabold}^{\mathrm{e}}) - \Lcal(\alphabold^*)\\
		& =  \Op\left( m L_q L_s \sqrt{\frac{\log(K)}{n}}  + M \sqrt{\frac{K \log(n)}{n}}\right),
	\end{align*}
	where the last equality follows from \refeq{eq:step_1_2_link} and \refeq{eq:first_result}.
\end{proof}

\begin{proof}[Proof of Lemma~\ref{lemma:quantile}]
	Using the Dvoretzky–Kiefer–Wolfowitz (DKW) inequality, we have
	\[\P\Big( \sup_{\alpha \in [0,1]} \Bigabs{\frac{1}{n} \sum_{i \in[n]} \indicator\{x_i \leq q(\alpha) \} - \alpha} \leq \epsilon \Big) \geq 1 - 2 \exp(-2 n \epsilon^2),\]
	where the continuity of the CDF $F$ ensures $F(q(\alpha)) = \alpha$. Setting $\epsilon = \sqrt{\log(2/\delta)/(2n)}$ yields that the event occurs with probability at least $1 - \delta$.
	Conditioning on this event, we obtain
	$$\frac{1}{n} \sum_{i \in[n]} \indicator\{x_i \leq q(\alpha) \} \leq \alpha +\epsilon.$$ 
	Let $\widehat{q}(\alpha)$ be the quantile function of $\{x_i\}_{i \in [n]}$, it follows that
	\[\widehat{q}(\alpha + \epsilon) \geq q(\alpha) \text{ for } \alpha \in [0,1 - \epsilon] ~~~~\Rightarrow~~~~ \widehat{q}(\alpha ) \geq q(\alpha - \epsilon) \text{ for } \alpha \in [\epsilon,1]. \]
	Similarly, we have
	\[\widehat{q}(\alpha ) \leq q(\alpha + \epsilon) \text{ for } \alpha \in [0,1-\epsilon].\]
	Combining both bounds gives, for $\alpha \in [\epsilon, 1 - \epsilon]$,
	\begin{align}\label{eq_quantile}
		q(\alpha - \epsilon) \leq \widehat{q}(\alpha) \leq q(\alpha + \epsilon).
	\end{align}
	
	For any $\alpha \in [a + \epsilon, b - \epsilon] \subset [\epsilon, 1 - \epsilon]$, the Lipschitz continuity of $q$ over $[a,b]$ gives
	\begin{align*}
		\max\big(\abs{q(\alpha) - q(\alpha + \epsilon)},\abs{q(\alpha) - q(\alpha - \epsilon)}\big) \leq L_q \epsilon.
	\end{align*}
	Combining with \refeq{eq_quantile}, we conclude that
	\[\abs{\widehat{q}(\alpha) - q(\alpha)} \leq L_q \epsilon,\]
	which holds uniformly over $\alpha \in [a + \epsilon, b - \epsilon]$. Noting that $\widehat{q}^+(\alpha) = \widehat{q}(\frac{n+1}{n}\alpha)$ for $\alpha \leq n/(n+1)$, we obtain the following bound, valid uniformly over $\alpha \in [a + \epsilon, \frac{n}{n+1}(b - \epsilon)  ]$:	
	\begin{align*}
		\widehat{q}^+(\alpha) - q(\alpha) & = \widehat{q}\Big(\frac{n+1}{n}\alpha\Big) - q(\alpha)\\
		& = \widehat{q}\Big(\frac{n+1}{n}\alpha\Big) - q\Big(\frac{n+1}{n}\alpha\Big) + q\Big(\frac{n+1}{n}\alpha\Big) - q(\alpha)\\
		& \leq \sup_{\alpha \in [a + \epsilon, b - \epsilon]} \abs{\widehat{q}(\alpha) - q(\alpha)} + q\Big(\frac{n+1}{n}\alpha\Big) - q(\alpha)\\
		& \leq L_q  \sqrt{\frac{\log(2/\delta)}{2n}} + \frac{L_q}{n},
	\end{align*}
	which concludes the proof.
\end{proof}

\begin{proof}[Proof of Lemma~\ref{lemma:sym_dif}]
	With loss of generality, we define $[a,b] = \varnothing$ if $a > b$.
	
	(i) 
	If there exists a $j$ such that $A_j \cap B_j = \varnothing$, then 
	$$\abs{\{\cap_{k = 1 }^K A_k \} \triangle\{\cap_{k = 1 }^K B_k \} } \leq \abs{A_j} + \abs{B_j} = \abs{A_j \triangle B_j} \leq c$$
	and the result follows immediately. Therefore, we assume $A_k \cap B_k \neq \varnothing$ for all $k$.
	
	Note that
	\begin{align*}
		\cap_{k = 1}^K A_k = [\max_k a_{k,l} , \min_k a_{k,r}],\ \cap_{k = 1}^K B_k = [\max_k b_{k,l} , \min_k b_{k,r}].
	\end{align*} 
	Since $\abs{A_k \triangle B_k} \leq c$ for $k\in\lrm{K}$, we have $\max_{k \in [K]}\max\{\abs{a_{k,l} - b_{k,l}} , \abs{a_{k,r} - b_{k,r}}\} \leq c$.
	If $\min_k a_{k,r} \geq \min_k b_{k,r}$, we have
	\[0 \leq \min_k a_{k,r} - \min_k b_{k,r} \leq \min_k a_{k,r} - \min_k \{a_{k,r} - c\} \leq c.\]
	If $\min_k a_{k,r} < \min_k b_{k,r}$, we have
	\[0 \leq \min_k b_{k,r} - \min_k a_{k,r} \leq \min_k \{a_{k,r} + c\} - \min_k a_{k,r} \leq c,\]
	which implies
	\[\abs{\min_k a_{k,r} - \min_k b_{k,r}} \leq c.\]
	Similarly, we can obtain
	\[\abs{\max_k a_{k,l} - \max_k b_{k,l}} \leq c.\]
	When both $\cap_{k = 1 }^K A_k\neq\varnothing$ and $\cap_{k = 1 }^K B_k\neq\varnothing$, the symmetric difference is bounded by
	\begin{align*}
		\abs{\{\cap_{k = 1 }^K A_k \} \triangle\{\cap_{k = 1 }^K B_k \} } \leq \abs{\max_k a_{k,l} - \max_k b_{k,l}} + \abs{\min_k a_{k,r} - \min_k b_{k,r}} \leq 2c.
	\end{align*} When $\cap_{k = 1 }^K A_k=\varnothing$ but $\cap_{k = 1 }^K B_k\neq\varnothing$, the bound becomes
	\begin{align*}
		\abs{\{\cap_{k = 1 }^K A_k \} \triangle\{\cap_{k = 1 }^K B_k \} }=\min_k b_{k,r} -\max_k b_{k,l} \leq \min_k b_{k,r} - \min_k a_{k,r} + \max_k a_{k,l} - \min_k b_{k,l} \leq 2c.
	\end{align*} The remaining cases can be handled similarly, which concludes the proof.
	
	(ii) By augmenting the unions with additional empty sets, or by decomposing longer intervals into multiple shorter subintervals, we can always ensure that there exist alternative representations for $A_k = \cup_{i = 1}^{m_{k}}[a_{k,i,l}, a_{k,i,r}]$ and $B_k = \cup_{i = 1}^{m_{k}}[b_{k,i,l}, b_{k,i,r}]$ such that $$\bigabs{[a_{k,i,l}, a_{k,i,r}] \triangle [b_{k,i,l}, b_{k,i,r}]} \leq c $$ for some $m_k \leq m_{a,k} +  m_{b,k}$.
	
	With loss of generality, assume $\cap_{k = 1}^K A_k = \cup_{j = 1}^m [a_{j,l}, a_{j,r}]$ for some $m \leq \max_k m_{a,k}$. For each $j$, there exists $i_{k,j} \leq m_k$ such that
	\[[a_{j,l},a_{j,r}] = \cap_{k = 1}^K [a_{k,i_{k,j},l}, a_{k,i_{k,j},r}].\]
	Let $[b_{j,l},b_{j,r}] = \cap_{k = 1}^K [b_{k,i_{k,j},l}, b_{k,i_{k,j},r}] \subset \cap_{k = 1}^K B_k$. From the first result (i), we have 
	$$\bigabs{[a_{j,l},a_{j,r}] \setminus [b_{j,l},b_{j,r}]} \leq \bigabs{[a_{j,l},a_{j,r}] \triangle [b_{j,l},b_{j,r}]} \leq 2c.$$
	This implies
	\begin{align*}
		\bigabs{\cap_{k = 1}^K A_k \setminus \cap_{k = 1}^K B_k} \leq \bigabs{\cup_{j = 1}^m [a_{j,l}, a_{j,r}] \setminus \cup_{j = 1}^m [b_{j,l}, b_{j,r}]} \leq 2 m c.
	\end{align*}
	Similarly, we have
	\begin{align*}
		\bigabs{\cap_{k = 1}^K B_k \setminus \cap_{k = 1}^K A_k} \leq 2 m c,
	\end{align*}
	which indicates
	\[\bigabs{\cap_{k = 1}^K A_k \triangle \cap_{k = 1}^K B_k} \leq 4 m c .\]
\end{proof}

\subsection{Proof of Theorem \texorpdfstring{\ref{thm:cola_s_coverage}}{3}}
\begin{proof}
	Marginalizing the following inequality over $\Dcal_{\mathrm{tu}}$ directly yields the desired result:
	\begin{align*}
		\P\Big(Y_{n+1} \in \cap_{k = 1}^K \widehat{C}_{k,\mathrm{cal}}(\mX_{n+1};\widehat{\alpha}_{k}^{\mathrm{s}}) \Big| \Dcal_{\mathrm{tu}} \Big) \geq 1 - \sum_k \P\Big(Y_{n+1} \notin \widehat{C}_k(\mX_{n+1};\widehat{\alpha}_{k}^{\mathrm{s}})  \Big| \Dcal_{\mathrm{tu}} \Big) \geq 1 - \alpha.
	\end{align*}
\end{proof}

\subsection{Proof of Corollary \texorpdfstring{\ref{thm:optimal_cola_s}}{1} } \label{proof:optimal_cola_s}

\begin{proof}[Proof of Corollary~\ref{thm:optimal_cola_s}]
	Observe that
	\begin{align*}
		\E_{\mX_{n+1}} \!\left[ \big|\widehat{C}^{\mathrm{s}}(\mX_{n+1};\alpha)\big| \right] - \Lcal(\alphabold^*)
		&= \E_{\mX_{n+1}}\!\left[ \Big|\!\cap_{k=1}^K \widehat{C}_{k,\mathrm{cal}} (\mX_{n+1};\widehat{\alpha}_{k}^\mathrm{s})\!\Big| - \Big|\!\cap_{k=1}^K C_k(\mX_{n+1};\alpha_k^*)\!\Big| \right] \\
		&= \E_{\mX_{n+1}}\!\left[ \Big|\!\cap_{k=1}^K \widehat{C}_{k,\mathrm{cal}} (\mX_{n+1};\widehat{\alpha}_{k}^\mathrm{s})\!\Big| - \Big|\!\cap_{k=1}^K C_k(\mX_{n+1};\widehat{\alpha}_{k}^\mathrm{s})\!\Big| \right] \\
		&~~~~ + \E_{\mX_{n+1}}\!\left[ \Big|\!\cap_{k=1}^K C_k(\mX_{n+1};\widehat{\alpha}_{k}^\mathrm{s})\!\Big| - \Big|\!\cap_{k=1}^K C_k(\mX_{n+1};\alpha_k^*)\!\Big| \right] \\
		&\le \Pi_1^{\mathrm{s}} + \Lcal(\widehat{\alphabold}^{\mathrm{s}}) - \Lcal(\alphabold^*),
	\end{align*}
	where
	\[
	\Pi_1^{\mathrm{s}} := \sup_{\alphabold, \vx} \Big| \{\cap_{k=1}^K \widehat{C}_{k,\mathrm{cal}}(\vx;\alpha_k)\} \triangle \{\cap_{k=1}^K C_k(\vx;\alpha_k)\} \Big|.
	\]
	By the argument in Theorem~\ref{thm:asymptotic_opt_cola_e} in Section~\ref{proof:eff_optimal} (with $\Dcal$ replaced by $\Dcal_{\mathrm{tu}}$),
	\[
	\Lcal(\widehat{\alphabold}^{\mathrm{s}}) - \Lcal(\alphabold^*) = \Op\!\left( m L_q L_s \sqrt{\frac{\log K}{n_{\mathrm{tu}}}} + M \sqrt{\frac{K \log n_{\mathrm{tu}}}{n_{\mathrm{tu}}}} \right).
	\]
	The term $\Pi_1^{\mathrm{s}}$ admits the same bound as $\Pi_1$ in~\eqref{eq:step_1_2_link}, namely
	\[
	\Pi_1^{\mathrm{s}} = O\!\left( m L_q L_s \sqrt{\frac{\log(2K/\delta)}{2n_{\mathrm{cal}}}} + \frac{m L_q L_s}{n_{\mathrm{cal}}} \right) \quad\text{ with probability at least}\  1-\delta.
	\]
	Hence,
	\[
	\E_{\mX_{n+1}} \!\left[ \big|\widehat{C}^{\mathrm{s}}(\mX_{n+1};\alpha)\big| \right] - \Lcal(\alphabold^*) 
	= \Op\!\left( m L_q L_s \sqrt{\frac{\log K}{n_{\mathrm{min}}}} + M \sqrt{\frac{K \log n_{\mathrm{min}}}{n_{\mathrm{min}}}} \right).
	\]
\end{proof}

\subsection{Proof of Theorem \texorpdfstring{\ref{thm:valid_full}}{4} } \label{proof:valid_full}
\begin{proof}[Proof of Theorem~\ref{thm:valid_full}]
	\begin{align*}
		\P\lrl{Y_{n+1}\in\widehat{C}^{\mathrm{f}}\lrs{\mX_{n+1}}} =&\P\lrl{Y_{n+1}\in\cap_{k=1}^K\widehat{C}_k^{(Y_{n+1})}\lrs{\mX_{n+1},\widehat{\alpha}_k^{(Y_{n+1})}}}\\
		\geq&1-\sum_{k\in[K]} \P\lrl{Y_{n+1}\notin\widehat{C}_k^{(Y_{n+1})}\lrs{\mX_{n+1},\widehat{\alpha}_k^{(Y_{n+1})}} }\\
		\geq& 1-\sum_{k\in[K]}\E\lrs{\widehat{\alpha}_k^{(Y_{n+1})}}
		\geq 1-\alpha,
	\end{align*}
	where the second-to-last inequality follows from\begin{align*}
		&\P\lrl{Y_{n+1}\notin\widehat{C}_k^{(Y_{n+1})}\lrs{\mX_{n+1},\widehat{\alpha}_k^{(Y_{n+1})}}}\\
		=&\P\lrl{S_k\lrs{\mX_{n+1},Y_{n+1}}> Q_{\widehat{\alpha}_k^{(Y_{n+1})}}(\{S_{k,i}\}_{i \in [n]} \cup \{S_k(\mX_{n+1},Y_{n+1})\})}\\
		=&\dfrac{1}{n+1}\sum_{i\in[n+1]}\P\lrl{S_k\lrs{\mX_{i},Y_{i}}> Q_{\widehat{\alpha}_k^{(Y_{n+1})}}(\{S_{k,i}\}_{i \in [n]} \cup \{S_k(\mX_{n+1},Y_{n+1})\})}\\
		=& \E\lrm{\dfrac{1}{n+1}\sum_{i\in[n+1]}\indicator\lrl{S_k\lrs{\mX_{i},Y_{i}}> Q_{\widehat{\alpha}_k^{(Y_{n+1})}}(\{S_{k,i}\}_{i \in [n]} \cup \{S_k(\mX_{n+1},Y_{n+1})\})}}\\
		\leq &\E\lrl{\dfrac{1}{n+1}\lrs{n+1-\Big\lceil\lrs{1-\widehat{\alpha}^{(Y_{n+1})}_k}\lrs{n+1}\Big\rceil}}\\
		\leq &\E\lrs{\widehat{\alpha}_k^{(Y_{n+1})}}
	\end{align*}
	where the second equality holds since $\widehat{\alpha}^{(Y_{n+1})}$ is a symmetric function of the $n+1$ data points $\lrs{\mX_1,Y_1},\dots,\lrs{\mX_{n+1},Y_{n+1}}$, and $\lrs{\mX_1,Y_1},\dots,\lrs{\mX_{n+1},Y_{n+1}}$ are exchangeable.
\end{proof}

\subsection{Proof of Theorem \texorpdfstring{\ref{thm:valid_conditional}}{5} }
The following lemma is needed for the proof of Theorem~\ref{thm:valid_conditional}. Proof of this lemma is provided at the end of this section. For simplicity, we denote $\widehat{q}_{k, \alpha}^{\mathrm{loc}} = Q_{\alpha} (\{ S_{k,i} \}_{i = 1}^n , \vw)$. For brevity, we denote $\widehat{\alphabold} = \widehat{\alphabold}(\mX_{n+1})$ and $\widehat{\alpha}_k = \widehat{\alpha}_k(\mX_{n+1})$ in the proofs of Theorems~\ref{thm:valid_conditional} and \ref{thm:optimal_conditional}. The Laplace kernel, given by $H(\vx,\vx')=\exp{\lrl{-\|\vx-\vx'\|/h_n}}$ is used in the proof; the argument for the Gaussian kernel follows analogously. 
\begin{lemma}\label{lemma:bound_on_weights}
	Let $\mathfrak{B} = \sum_{i\in[n]} H(\mX_i, \mX_{n+1})$. If the density function of $ \mX_{n+1} $ is lower bounded by $\rho$ and $\sqrt{\log(1/\delta)/(n \rho V_d h_n^d)} < 1 - c$ for some constant $0<c<1$, then conditional on $\mX_{n+1}$, the following events hold
	\begin{align*}
		\mathfrak{B}  \geq \frac{c n \rho V_d h_n^d}{e},~~~
		\sum_{i \in [n]} w_i^2  \leq \frac{e}{c n \rho V_d h_n^d},
	\end{align*}
	with probability at least $1-\delta$. Here, $V_d$ is the volume of the unit ball in $\real^d$.
\end{lemma}
\label{proof:valid_conditional}
\begin{proof}[Proof of Theorem~\ref{thm:valid_conditional}]	
	First, we prove Part~(1).  
	Note that
	\begin{equation}\label{eq:conditional_valid_main}
		\begin{aligned}
			\P\!\left( Y_{n+1} \in \widehat{C}^{\mathrm{loc}}(\mX_{n+1}) \,\middle|\, \Dcal, \mX_{n+1} \right)
			&= \P\!\left( Y_{n+1} \in \cap_{k=1}^K \widehat{C}_{k}^{\mathrm{loc}}(\mX_{n+1}; \widehat{\alpha}_k) \,\middle|\, \Dcal, \mX_{n+1} \right) \\
			&\geq 1 - \sum_{k\in[K]} \P\!\left( S_k(\mX_{n+1}, Y_{n+1}) > \widehat{q}_{k,\widehat{\alpha}_k}^{\mathrm{loc}} \,\middle|\, \Dcal, \mX_{n+1} \right) \\
			&= 1 - \sum_{k\in[K]} \widehat{\alpha}_k - K \Xi = 1 - \alpha - K \Xi,
		\end{aligned}
	\end{equation}
	where
	\begin{equation}\label{eq:conditional_valid_error}
		\begin{aligned}
			\Xi := & \sup_{k \in [K]} \sup_{\alpha \in [0,1]} 
			\left[ \P\!\left( S_k(\mX_{n+1}, Y_{n+1}) > \widehat{q}_{k,\alpha}^{\mathrm{loc}} \,\middle|\, \Dcal, \mX_{n+1} \right) - \alpha \right]\\
			\leq\ & \sup_{k \in [K]} \sup_{\alpha \in [0,1]} 
			\left[ \P\!\left( S_k(\mX_{n+1}, Y_{n+1}) > \widehat{q}_{k,\alpha}^{\mathrm{loc}} \,\middle|\, \Dcal, \mX_{n+1} \right) 
			- \sum_{i\in[n]} w_i \,\indicator\!\left\{ S_k(\mX_i, Y_i) > \widehat{q}_{k,\alpha}^{\mathrm{loc}} \right\} \right]\\
			\leq\ & \sup_{k \in [K]} \sup_{\alpha \in [0,1]} 
			\left| \P\!\left( S_k(\mX_{n+1}, Y_{n+1}) \leq \widehat{q}_{k,\alpha}^{\mathrm{loc}} \,\middle|\, \Dcal, \mX_{n+1} \right) 
			- \sum_{i\in[n]} w_i \,\indicator\!\left\{ S_k(\mX_i, Y_i) \leq \widehat{q}_{k,\alpha}^{\mathrm{loc}} \right\} \right| \\
			\leq\ & \underbrace{\sup_{k \in [K]} \sup_{q \in \mathbb{R}} 
				\left| \P\!\left( S_k(\mX_{n+1}, Y_{n+1}) \leq q \,\middle|\, \Dcal, \mX_{n+1} \right)
				- \sum_{i\in[n]} w_i \,\P\!\left( S_k(\mX_i, Y_i) \leq q \,\middle|\, \mX_i \right) \right|}_{\Xi_1} \\
			&\quad + \underbrace{\sup_{k \in [K]} \sup_{q \in \mathbb{R}}
				\left| \sum_{i\in[n]} w_i \,\P\!\left( S_k(\mX_i, Y_i) \leq q \,\middle|\, \mX_i \right) 
				- \sum_{i\in[n]} w_i \,\indicator\!\left\{ S_k(\mX_i, Y_i) \leq q \right\} \right|}_{\Xi_2} ,\\
		\end{aligned}
	\end{equation}
	where the first inequality follows from the definition of the quantile.
	
	First, we bound the term $\Xi_1$. Under the assumption that $ \sup_{k \in [K]} \sup_{q \in \mathbb{R}} 
	\big| F_{k}(q|\vx) - F_{k}(q|\vx') \big|
	\leq \tau \,\|\vx - \vx'\|$, we have
	\begin{align*}
		&~~~~\bigg| 
		\P\!\left( S_k(\mX_{n+1}, Y_{n+1}) \leq q \,\middle|\, \Dcal, \mX_{n+1} \right)
		- \sum_{i\in[n]} w_i \P\!\left( S_k(\mX_i, Y_i) \leq q \,\middle|\, \mX_i \right)
		\bigg| \\
		&\leq \sum_{i\in[n]} w_i \,\tau \,\|\mX_i - \mX_{n+1}\| \\
		&= \mathfrak{B}^{-1} \sum_{j=1}^\infty \sum_{i\in[n]} 
		H(\mX_i, \mX_{n+1}) \,\tau \,\|\mX_i - \mX_{n+1}\| 
		\,\indicator\{\mX_i \in \Xcal_j\} \\
		&\leq \mathfrak{B}^{-1} \sum_{j=1}^\infty \sum_{i\in[n]} 
		e^{-j+1} \,\tau \,j h_n \,\indicator\{\mX_i \in \Xcal_j\} \\
		&= \mathfrak{B}^{-1} \sum_{j\in[j_0]} \sum_{i\in[n]} 
		e^{-j+1} \,\tau \,j h_n \,\indicator\{\mX_i \in \Xcal_j\} + \mathfrak{B}^{-1} \sum_{j=j_0+1}^\infty \sum_{i\in[n]} 
		e^{-j+1} \,\tau \,j h_n \,\indicator\{\mX_i \in \Xcal_j\} \\
		&\leq j_0 h_n \tau \,\mathfrak{B}^{-1} 
		\sum_{i\in[n]} \sum_{j\in[j_0]} e^{-j+1} \,\indicator\{\mX_i \in \Xcal_j\} + h_n \tau \,\mathfrak{B}^{-1} n \sum_{j=j_0+1}^\infty j e^{-j+1} \\
		&\leq e \,j_0 h_n \tau \,\mathfrak{B}^{-1} \mathfrak{B} 
		+ \Op\!\left( h_n \tau \,\mathfrak{B}^{-1} n \, j_0 e^{-j_0} \right) \\
		&= e\, j_0 h_n \tau 
		+ \Op\!\left( h_n \tau \,\mathfrak{B}^{-1} n \, j_0 e^{-j_0} \right),
	\end{align*}
	where $ \Xcal_j = \left\{ \vx : j-1 \leq \frac{\|\vx - \mX_{n+1}\|}{h_n} < j \right\} $, the second inequality uses  
	\[
	H(\mX_i, \mX_{n+1}) \indicator\{\mX_i \in \Xcal_j\} 
	\leq e^{-j+1} \indicator\{\mX_i \in \Xcal_j\}, 
	\quad
	\|\mX_i - \mX_{n+1}\| \indicator\{\mX_i \in \Xcal_j\} \leq j h_n,
	\]
	the second equality holds for any $j_0 \in \mathbb{Z}^+$,  
	and the last inequality follows from  
	$H(\mX_i, \mX_{n+1}) \indicator\{\mX_i \in \Xcal_j\} \geq e^{-j} \indicator\{\mX_i \in \Xcal_j\}$. Choosing $j_0 = O\big( \log(h_n^{-d}) \big)$ and using Lemma~\ref{lemma:bound_on_weights},  
	we obtain
	\begin{equation}\label{eq:conditional_valid_term1}
		\Xi_1 = O\!\left( \frac{\tau h_n \log(h_n^{-d})}{\rho} \right)
	\end{equation}
	with probability at least $1 - \delta$.
	
	Now, we proceed to bound the term $\Xi_2$. For simplicity, denote $F_{k,i} (q)= F_{k}(q|\mX_i)$. Under the assumption that $F_{k}(q|\vx)$ is continuous, we have for each $k$, 
	\begin{align*} 
		& \sup_{q \in \real} \bigabs{\sum_{i\in[n] } w_i \P\lrl{ S_{k,i} \leq  q | \mX_{i}} - \sum_{i\in[n]} w_i \indicator \{S_k(\mX_i, Y_i) \leq q\}}\\
		= & \sup_{q \in \real} \bigabs{\sum_{i\in[n]} w_i \P\lrl{ F_{k,i}(S_{k,i}) \leq  F_{k,i}(q) | \mX_{i}} - \sum_{i\in[n]} w_i \indicator \{F_{k,i}(S_{k,i}) \leq  F_{k,i}(q)\}}\\
		= & \sup_{q \in [0,1]} \bigabs{\sum_{i\in[n]} w_i \P\lrl{ F_{k,i}(S_{k,i}) \leq  q | \mX_{i}} - \sum_{i\in[n]} w_i \indicator \{F_{k,i}(S_{k,i}) \leq  q\}}\\
		= & \sup_{q \in [0,1]} \bigabs{q - \sum_{i\in[n]} w_i \indicator \{F_{k,i}(S_{k,i}) \leq  q\}},
	\end{align*}
	where the last equality comes from the fact that $F_{k,i}(S_{k,i})|\mX_i \sim \mathrm{Unif}[0,1]$. Partition the interval $[0,1]$ into $m_0$ equally spaced subintervals $[q_\ell,q_{\ell+1}]$ for $\ell \in [m_0]$ with $q_1 = 0$ and $q_{m_0+1} = 1$. If $q \in [q_\ell,q_{\ell+1}]$, then
	\begin{align*} 
		& \bigabs{q - \sum_{i\in[n]} w_i \indicator \{F_{k,i}(S_{k,i}) \leq  q\}}\\
		\leq & \bigabs{q_\ell - \sum_{i\in[n]} w_i \indicator \{F_{k,i}(S_{k,i}) \leq  q_\ell\}} + \abs{q - q_\ell} + \sum_{i\in[n]} w_i \indicator \{q_\ell < F_{k,i}(S_{k,i}) \leq  q\}\\
		\leq & \bigabs{q_{\ell} - \sum_{i\in[n]} w_i \indicator \{F_{k,i}(S_{k,i}) \leq  q_\ell\}} + \frac{1}{m_0} + \sum_{i\in[n]} w_i \indicator \{q_\ell < F_{k,i}(S_{k,i}) \leq  q_{\ell+1}\}\\
		\leq & \bigabs{q_{\ell} - \sum_{i\in[n]} w_i \indicator \{F_{k,i}(S_{k,i}) \leq  q_\ell\}} + \frac{1}{m_0} \\
		& \quad + \bigabs{\sum_{i\in[n]} w_i \indicator \{q_\ell < F_{k,i}(S_{k,i}) \leq  q_{\ell+1}\} - ( q_{\ell+1} -  q_\ell)} + q_{\ell+1} -  q_\ell\\
		\leq & \bigabs{q_{\ell} - \sum_{i\in[n]} w_i \indicator \{F_{k,i}(S_{k,i}) \leq  q_\ell\}} + \frac{2}{m_0} \\
		& \quad + \bigabs{\sum_{i\in[n]} w_i \indicator \{q_\ell < F_{k,i}(S_{k,i}) \leq  q_{\ell+1}\} - ( q_{\ell+1} -  q_\ell)}.
	\end{align*}
	It follows that
	\begin{align*}
		& ~~~~\sup_{q \in [0,1]} \bigabs{q - \sum_{i\in[n]} w_i \indicator \{F_{k,i}(S_{k,i}) \leq  q\}} \\
		& = \sup_{\ell \in [m_0]} \sup_{q \in [q_\ell,q_{\ell+1}]} \bigabs{q - \sum_{i\in[n]} w_i \indicator \{F_{k,i}(S_{k,i}) \leq  q\}}\\
		& \leq \sup_{\ell \in [m_0]} \bigabs{q_{\ell} - \sum_{i\in[n]} w_i \indicator \{F_{k,i}(S_{k,i}) \leq  q_\ell\}}+ \frac{2}{m_0}\\
		& \quad + \sup_{\ell \in [m_0]} \bigabs{\sum_{i\in[n]} w_i \indicator \{q_\ell < F_{k,i}(S_{k,i}) \leq  q_{\ell+1}\} - ( q_{\ell+1} -  q_\ell)}.
	\end{align*}
	For any fixed $\ell \in [m_0]$, Hoeffding's inequality gives
	\begin{align*}
		& \P\Bigg( \bigabs{q_{\ell} - \sum_{i\in[n]} w_i \indicator \{F_{k,i}(S_{k,i}) \leq  q_\ell\}} \geq t \sqrt{\sum_{i\in[n]} w_i^2}  \Bigg| \{\mX_i\}_{i \in [n+1]}\Bigg) \leq 2 \exp(-2t^2),\\
		& \P\Bigg( \bigabs{\sum_{i\in[n]} w_i \indicator \{q_\ell < F_{k,i}(S_{k,i}) \leq  q_{\ell+1}\} - ( q_{\ell+1} -  q_\ell)} \geq t \sqrt{\sum_{i\in[n]} w_i^2}  \Big| \{\mX_i\}_{i \in [n+1]}\Bigg) \leq 2 \exp(-2t^2).
	\end{align*}
	Taking the union bound over $k \in [K]$ and $\ell \in [m_0]$ and marginalizing over $\{\mX_i: i \in [n]\}$ yields
	\begin{align*}
		& \P\Bigg(\sup_{k \in [K]} \sup_{\ell \in [m_0]} \bigabs{q_{\ell} - \sum_{i\in[n]} w_i \indicator \{F_{k,i}(S_{k,i}) \leq  q_\ell\}} \geq t \sqrt{\sum_{i\in[n]} w_i^2}  \Bigg| \mX_{n+1}\Bigg) \leq 2 m_0 K \exp(-2t^2),\\
		& \P\Bigg(\sup_{k \in [K]} \sup_{\ell \in [m_0]} \bigabs{\sum_{i\in[n]} w_i \indicator \{q_\ell < F_{k,i}(S_{k,i}) \leq  q_{\ell+1}\} - ( q_{\ell+1} -  q_\ell)} \geq t \sqrt{\sum_{i\in[n]} w_i^2}  \Big| \mX_{n+1} \Bigg) \\
		& \quad \leq 2 m_0 K \exp(-2t^2).
	\end{align*}
	Applying Lemma~\ref{lemma:bound_on_weights} and taking $m_0 = O\Big(\sqrt{\frac{n \rho V_d h_n^d}{\log(K)}} \Big)$, we have
	\begin{align}\label{eq:conditional_valid_term2}
		\Xi_2 = O\Big(\sqrt{\frac{\log(K n \rho V_d h_n/\delta)}{n \rho V_d h_n^d}}\Big) ,
	\end{align}
	with probability at least $1-\delta$.
	
	Combining \refeq{eq:conditional_valid_error}, \refeq{eq:conditional_valid_term1}, and \refeq{eq:conditional_valid_term2}, we obtain
	\begin{align*}
		\Xi = O\Bigg(\frac{\tau h_n \log(h_n^{-d})}{\rho} + \sqrt{\frac{\log(K n \rho V_d h_n/\delta)}{n \rho V_d h_n^d}} \Bigg),
	\end{align*}
	with probability at least $ 1- 2 \delta $. Further, we can deduce that 
	\begin{align*}
		\E[\Xi|\mX_{n+1}] = O\Bigg(\frac{\tau h_n \log(h_n^{-d})}{\rho} + \sqrt{\frac{\log(K n \rho V_d h_n)}{n \rho V_d h_n^d}} \Bigg).
	\end{align*}
	Substituting this bound into \refeq{eq:conditional_valid_main} and marginalizing over $ \Dcal $, we conclude Part (1) via
	\begin{align*}
		\P\lrl{Y_{n+1}\in\widehat{C}^{\mathrm{loc}}\lrs{\mX_{n+1}} \mid \mX_{n+1}} 
		& \geq 1-\alpha - K \cdot \E[\Xi|\mX_{n+1}]\\
		& \geq 1-\alpha - K \cdot O \Bigg(\frac{\tau h_n \log(h_n^{-d})}{\rho} + \sqrt{\frac{\log(K n \rho V_d h_n)}{n \rho V_d h_n^d}} \Bigg).
	\end{align*}
	
	Now, we proceed to prove Part (2). Note that
	\begin{equation*}
		\begin{aligned}
			& \P\Big(Y_{n+1} \in \cap_k \widehat{C}_{k}^\mathrm{loc} (\mX_{n+1};\widehat{\alpha}_k) \big|  \mX_{n+1}, \Dcal \Big) 
			- \P\Big(Y_{n+1} \in \cap_k C_{k}^\mathrm{loc} (\mX_{n+1};\widehat{\alpha}_k) \big|  \mX_{n+1}, \Dcal \Big)\\
			\geq & - \sup_{\alphabold} \P\Big(Y_{n+1} \in \cap_k \widehat{C}_{k}^\mathrm{loc} (\mX_{n+1};\alpha_k) \triangle \cap_k C_{k}^\mathrm{loc} (\mX_{n+1};\alpha_k) \big|  \mX_{n+1}, \Dcal \Big) \\ 
			\geq & - M_f \sup_{\alphabold} \Bigabs{ \cap_k \widehat{C}_{k}^\mathrm{loc} (\mX_{n+1};\alpha_k) \triangle \cap_k C_{k}^\mathrm{loc} (\mX_{n+1};\alpha_k)},
		\end{aligned}
	\end{equation*}
	where the second inequality follows from the upper bound on the conditional density of $Y_{n+1}|\mX_{n+1}$. Combining the above inequality with 
	\[
	\P\Big(Y_{n+1} \in \cap_k C_{k}^\mathrm{loc} (\mX_{n+1};\widehat{\alpha}_k) \big|  \mX_{n+1}, \Dcal \Big) 
	\geq 1 - \sum_k \P\Big(Y_{n+1} \notin C_{k}^\mathrm{loc} (\mX_{n+1};\widehat{\alpha}_k) \big|  \mX_{n+1}, \Dcal \Big) \geq 1- \alpha,
	\]
	and marginalizing over $\Dcal$, we obtain
	\begin{equation*}
		\begin{aligned}
			&\P\Big(Y_{n+1} \in \cap_k \widehat{C}_{k}^\mathrm{loc} (\mX_{n+1};\widehat{\alpha}_k) \big|  \mX_{n+1} \Big)\\
			\geq & 1 - \alpha - M_f \E\left\{\sup_{\alphabold} \abs{ \cap_k \widehat{C}_{k}^\mathrm{loc} (\mX_{n+1};\alpha_k) \triangle \cap_k C_{k}^\mathrm{loc} (\mX_{n+1};\alpha_k)}\Big| \mX_{n+1} \right\}.
		\end{aligned}
	\end{equation*}
	Using the bound in \refeq{eq:condition_set_diff}, we obtain
	\begin{align*}
		& \E\left\{\sup_{\alphabold} \abs{ \cap_k \widehat{C}_{k}^\mathrm{loc} (\mX_{n+1};\alpha_k) \triangle \cap_k C_{k}^\mathrm{loc} (\mX_{n+1};\alpha_k)}\Big| \mX_{n+1} \right\} \\
		= & m L_q L_s ~O\Bigg(\sqrt{\frac{\log(K n \rho V_d h_n)}{n \rho V_d h_n^d}} + \frac{\tau h_n \log(h_n^{-d})}{\rho} \Bigg),
	\end{align*}
	which completes the proof of Part (2).
\end{proof}

\begin{proof}[Proof of Lemma~\ref{lemma:bound_on_weights}]
	Decompose the feature space as $\Xcal = \bigcup_{j = 1}^\infty \Xcal_j$, where
	\[
	\Xcal_j = \left\{ \vx : j-1 \leq \frac{\|\vx - \mX_{n+1}\|}{h_n} < j \right\}.
	\]
	Observe that
	\[
	H(\mX_i, \mX_{n+1}) 
	= \exp\!\left( - \frac{\|\mX_i - \mX_{n+1}\|}{h_n} \right)
	\geq e^{-1} \, \indicator\{\mX_i \in \Xcal_1\},
	\]
	which implies $ e \cdot \mathfrak{B} \ \geq\  \sum_{i\in[n]} \indicator\{\mX_i \in \Xcal_1\} $. 
	Let $p = \P(\mX_1 \in \Xcal_1 \mid \mX_{n+1})$.  
	By Chernoff bound (Theorem~4.5 in~\cite{mitzenmacher2017probability}), for any $t \in (0,1)$,
	\[
	\P\!\left( \sum_{i\in[n]} \indicator\{\mX_i \in \Xcal_1\} \leq (1-t) n p \,\Big|\, \mX_{n+1} \right) 
	\leq \exp\!\left( - \frac{n p t^2}{2} \right).
	\]
	Choosing $t = \sqrt{\log(1/\delta)/(n p)} < 1$, we obtain with probability at least $1 - \delta$ (conditional on $\mX_{n+1}$),
	\[
	\sum_{i\in[n]} \indicator\{\mX_i \in \Xcal_1\} 
	> (1-t) n p 
	\ \geq\ c\, n \rho V_d h_n^d,
	\]
	where the last inequality uses the assumption that the density of $\mX_i$ is lower bounded by $\rho$, so that $p \geq \rho V_d h_n^d$, and $ \sqrt{ \frac{\log(1/\delta)}{n \rho V_d h_n^d} } < 1 - c $.
	This directly yields the first result:
	\[
	\mathfrak{B} \ \geq\ \frac{c\, n \rho V_d h_n^d}{e}.
	\]
	
	Moreover, since $H(\mX_i, \mX_{n+1}) \leq 1$, it follows that
	\[
	\sum_{i\in[n]} w_i^2 
	= \sum_{i\in[n]}  \frac{H^2(\mX_i,\mX_{n+1})}{\mathfrak{B}^2} 
	\ \leq\ \sum_{i\in[n]}  \frac{H(\mX_i,\mX_{n+1})}{\mathfrak{B}^2} 
	= \frac{1}{\mathfrak{B}},
	\]
	which completes the proof of the second result.
\end{proof}

\subsection{Proof of Theorem \texorpdfstring{\ref{thm:optimal_conditional}}{6} }\label{proof:optimal_conditional}
\begin{proof}[Proof of Theorem~\ref{thm:optimal_conditional}]
	First, we derive the estimation error bound for the weighted sample quantile:
	\[
	\sup_{k \in [K]} \sup_{\beta \in [l_\alpha,\alpha]} \bigabs{\widehat{q}_{k, \beta}^{\mathrm{loc}} - q_{k,\beta}(\mX_{n+1})}.
	\]
	Let $ q^+_{k,\beta} = F_k^{-1}(1 - \beta + c_n \mid \mX_{n+1}) $ and $ q^-_{k,\beta} = F_k^{-1}(1 - \beta - c_n \mid \mX_{n+1}) $,
	where $c_n = o(1)$ is a vanishing sequence to be specified later. Then
	\begin{align*}
		& \P\Big(\forall k \in [K], \forall \beta \in [l_\alpha,\alpha], \widehat{q}_{k,\beta}^{\mathrm{loc}} \leq q^+_{k,\beta} \Big| \mX_{n+1} \Big)\\
		= & \P\Big((\forall k \in [K], \forall \beta \in [l_\alpha,\alpha], \sum_{i\in[n]} w_i \indicator\{S_{k,i} \leq q^+_{k,\beta}\} \geq 1 - \beta \Big| \mX_{n+1} \Big)\\
		= & \P\Big((\forall k \in [K], \forall \beta \in [l_\alpha,\alpha], \sum_{i\in[n]} w_i \indicator\{S_{k,i} \leq q^+_{k,\beta}\} - F_k(q^+_{k,\beta} \mid \mX_{n+1}) \geq -c_n \Big| \mX_{n+1} \Big)\\
		\geq & \P\Big(\sup_{k \in [K]} \sup_{\beta \in [l_\alpha,\alpha]} \Bigabs{\sum_{i\in[n]} w_i \indicator\{S_{k,i} \leq q^+_{k,\beta}\} - F_k(q^+_{k,\beta} \mid \mX_{n+1})} \leq c_n \Big| \mX_{n+1} \Big)\\
		\geq & \P\Big(\sup_{k \in [K]} \sup_{q \in \real} \Bigabs{\sum_{i\in[n]} w_i \indicator\{S_{k,i} \leq q\} - F_k(q \mid \mX_{n+1})} \leq c_n \Big| \mX_{n+1} \Big)\\
		\geq & \P\Big(\sup_{k \in [K]} \sup_{q \in \real } \Bigabs{\sum_{i\in[n]} w_i F_k(q|\mX_i) -  F_k(q|\mX_{n+1}) } \leq \frac{c_n}{2} \Big| \mX_{n+1} \Big)  \\
		& +  \P\Big(\sup_{k \in [K]} \sup_{q \in \real } \Bigabs{ \sum_{i\in[n]} w_i \indicator\{S_{k,i} \leq q\} - \sum_{i\in[n]} w_i F_k(q|\mX_i)} \leq \frac{c_n}{2} \Big| \mX_{n+1} \Big) - 1\\
		= & \P\Big(\Xi_1 \leq c_n/2 \mid \mX_{n+1}\Big) + \P\Big(\Xi_2 \leq c_n/2 \mid \mX_{n+1}\Big) - 1 \geq 1 - \delta/2,
	\end{align*}
	where $\Xi_1$ and $\Xi_2$ are defined in \refeq{eq:conditional_valid_error}, and the last inequality follows from \refeq{eq:conditional_valid_term1} and \refeq{eq:conditional_valid_term2} by choosing
	\[
	c_n = O\Bigg(\sqrt{\frac{\log(K n \rho V_d h_n / \delta)}{n \rho V_d h_n^d}} + \frac{\tau h_n \log(h_n^{-d})}{\rho} \Bigg).
	\] 
	Similarly, we have
	\[
	\P\Big(\forall k \in [K], \forall \beta \in [l_\alpha,\alpha], \widehat{q}_{k,\beta}^{\mathrm{loc}} \geq q^-_{k,\beta} \mid \mX_{n+1}\Big) \geq 1 - \delta/2.
	\]
	On the intersection of these two events, we have
	\[
	q^-_{k,\beta} \leq q_{k,\beta}(\mX_{n+1}) \leq q^+_{k,\beta}, \quad q^-_{k,\beta} \leq \widehat{q}_{k,\beta}^{\mathrm{loc}} \leq q^+_{k,\beta},
	\]
	and therefore,
	\[
	\abs{q_{k,\beta}(\mX_{n+1}) - \widehat{q}_{k,\beta}^{\mathrm{loc}}} \leq \max\big(\abs{q_{k,\beta}(\mX_{n+1}) - q^-_{k,\beta}}, \abs{q_{k,\beta}(\mX_{n+1}) - q^+_{k,\beta}} \big) \leq L_q c_n,
	\]
	where the last inequality follows from Assumption~\ref{ass:local_ass}(d).
	Under Assumption~\ref{ass:common_ass}(b), it follows
	\[
	\sup_{k \in [K]} \sup_{\beta \in [l_\alpha,\alpha]} \abs{\widehat{C}_k^{\mathrm{loc}}(\mX_{n+1};\beta) \triangle C_k^{\mathrm{loc}}(\mX_{n+1};\beta)} \leq L_q L_s c_n.
	\]
	For any $\alphabold \in \boldsymbol{\Theta}$, Assumption~\ref{ass:local_ass}(e) gives
	\[
	\Bigabs{\cap_{k = 1}^K C_k^{\mathrm{loc}}(\vx;\alpha_k) \triangle \cap_{k = 1}^K \widehat{C}_k^{\mathrm{loc}}(\vx;\alpha_k)}
	= \Bigabs{\cap_{k = 1}^K C_k^{\mathrm{loc}}(\vx;\alpha_k^\prime) \triangle \cap_{k = 1}^K \widehat{C}_k^{\mathrm{loc}}(\vx;\alpha_k^\prime)},
	\]
	where $\alpha_k^\prime = \alpha_k \indicator\{\alpha_k \geq l_\alpha\} + l_\alpha \indicator\{\alpha_k < l_\alpha\}$.
	Applying Assumption~\ref{ass:common_ass}(a) and Lemma~\ref{lemma:sym_dif} yields
	\begin{align}\label{eq:condition_set_diff}
		\sup_{\alphabold \in \boldsymbol{\Theta}} \Bigabs{\cap_{k = 1}^K \widehat{C}_k^{\mathrm{loc}}(\mX_{n+1};\alpha_k) \triangle \cap_{k = 1}^K C_k^{\mathrm{loc}}(\mX_{n+1};\alpha_k)} \leq \Op(L_q L_s m c_n).
	\end{align}
	Consequently,
	\begin{align*}
		& ~~~~\Bigabs{\cap_{k = 1}^K \widehat{C}_k^{\mathrm{loc}}(\mX_{n+1};\widehat{\alpha}_k(\mX_{n+1}))} 
		- \Bigabs{\cap_{k = 1}^K C_k^{\mathrm{loc}}(\mX_{n+1};\alpha_k^*(\mX_{n+1}))} \\
		& \leq \Bigabs{\cap_{k = 1}^K \widehat{C}_k^{\mathrm{loc}}(\mX_{n+1};\alpha_k^*(\mX_{n+1}))} 
		- \Bigabs{\cap_{k = 1}^K C_k^{\mathrm{loc}}(\mX_{n+1};\alpha_k^*(\mX_{n+1}))} \\
		& = \Op(L_q L_s m c_n),
	\end{align*}
	which completes the proof.
\end{proof}

\subsection{Theoretical Results for Classification}\label{sec:classification}
In this section, we study the theoretical properties of $\mathrm{COLA\text{-}e}$ in the classification setting. The coverage guarantee in Theorem~\ref{thm:valid_cola_e} holds under only the i.i.d. assumption, and thus applies to both regression and classification problems. Therefore, we concentrate here on the efficiency. While the lack of continuity assumptions leads to a formulation that differs slightly from Theorem~\ref{thm:asymptotic_opt_cola_e}, the core idea and conclusion remain essentially the same.

\begin{theorem}\label{thm:classification_cola_e}
	For a classification problem with label space $\Ycal = [M]$, under the i.i.d. assumption of the $\Dcal$ and the test point $(\mX_{n+1},Y_{n+1})$, the $\mathrm{COLA\text{-}e}$ prediction set satisfies
	\[\E_{\mX_{n+1}}\lrm{|\widehat{C}^{\mathrm{e}}(\mX_{n+1};\alpha)|} - \Lcal(\alphabold^* - W_f^\prime \cdot \vone_K) = \Op\left( M \sqrt{\frac{\log(Kn)}{n}} \right),\]
	where $W_f^\prime = \Op( \sqrt{\log(K)/n})$ and $\vone_K$ denotes a $K$-dimensional vector with all entries equal to one.
\end{theorem}

Under the only i.i.d. assumption, this theorem primarily shows that the expected size of the aggregated conformal prediction set is asymptotically as small as that of the oracle prediction set corresponding to the allocation $\alphabold^* - W_f^\prime \cdot \vone_K$. With some continuity assumption on $\Lcal(\alphabold)$, the approximation error can be made explicit, yielding a result similar to the second result in Theorem~\ref{thm:asymptotic_opt_cola_e}.

\begin{proof}[Proof of Theorem~\ref{thm:classification_cola_e}]
	We begin by introducing two empirical process deviation terms for any fixed $y \in \Ycal$:
	\begin{align*}
		W_y & = \sup_{\vt}\left|\frac{1}{n}\sum_{i\in[n]}\Pi_{k=1}^K \indicator\lrl{S_k(\mX_i,y)\leq t_k}-\E_{\mX_{n+1}}\lrm{\Pi_{k=1}^K \indicator\lrl{S_k(\mX_{n+1},y)\leq t_k}} \right|,\\
		W_f & = \sup_{t \in \real} \sup_{k \in [K]} \left| \frac{1}{n}\sum_{i\in[n]}\indicator\lrl{S_{k,i}\leq t}-\P\lrs{S_k(\mX_{n+1},Y_{n+1})\leq t}\right|.
	\end{align*}
	Let $ \widehat{q}_{k,\alpha} = Q_\alpha(\{S_{k,i}\}_{i \in [n]})$. By the definition of $\widehat{q}_{k,\alpha}^+ $, we have 
	\begin{align}\label{eq:diff_quan_infty}
		\widehat{q}_{k,\alpha}^+ =  \widehat{q}_{k,\frac{n+1}{n}\alpha - \frac{1}{n}}
	\end{align}
	for any $\alpha\in(0,1)$. By the the definition of $W_f$, it follows that
	\begin{equation}
		\begin{aligned}\label{eq:quantile_error}
			\widehat{q}_{k,\alpha} \leq \inf\lrl{t:\P\lrs{S_k(\mX_{n+1},Y_{n+1})\leq t} \geq (1 - \alpha+W_f)\wedge 1}=q_{k,\alpha - W_f}.
		\end{aligned}
	\end{equation}
	Without loss of generality, we define $q_{k, \alpha - W_f} := q_{k,0}$ if $\alpha < W_f$. 
	
	We now analyze the expected size of the estimated conformal set:
	\begin{align*}
		\E_{\mX_{n+1}}\lrm{|\widehat{C}^{\mathrm{e}}(\mX_{n+1};\alpha)|}
		&=\E_{\mX_{n+1}}\lrm{\sum_{y\in[M]}\Pi_{k=1}^K\indicator\lrl{ y\in\widehat{C}_k\lrs{\mX_{n+1};\widehat{\alpha}_k^{\mathrm{e}}}}}\\ 
		&\leq\E_{\mX_{n+1}}\lrm{\sum_{y\in[M]}\Pi_{k=1}^K\indicator\lrl{ y\in\widehat{C}_k\lrs{\mX_{n+1};\alpha^*_k}}}\\       
		&=\sum_{y\in[M]}\E_{\mX_{n+1}}\lrm{\Pi_{k=1}^K\indicator\lrl{ S_k(\mX_{n+1},y)\leq \widehat{q}^+_{k,\alpha^*_k }}}\\
		&=\sum_{y\in[M]}\E_{\mX_{n+1}}\lrm{\Pi_{k=1}^K\indicator\lrl{ S_k(\mX_{n+1},y)\leq \widehat{q}_{k,\frac{n+1}{n}\alpha^*_k - \frac{1}{n} }}}\\
		&\leq\frac{1}{n}\sum_{i\in[n]} \sum_{y\in[M]}\Pi_{k=1}^K \indicator\lrl{S_k(\mX_i,y)\leq \widehat{q}_{k,\frac{n+1}{n} \alpha_k^* - \frac{1}{n}}}+\sum_{y\in[M]}W_y\\
		&\leq \frac{1}{n}\sum_{i\in[n]} \sum_{y\in[M]}\Pi_{k=1}^K \indicator\lrl{S_k(\mX_i,y)\leq q_{k,\frac{n+1}{n} \alpha_k^* - \frac{1}{n}-W_f}}+\sum_{y\in[M]}W_y\\
		&\leq \sum_{y\in[M]}\E_{\mX_{n+1}}\lrm{\Pi_{k=1}^K \indicator\lrl{S_k(\mX_{n+1},y)\leq  q_{k,\frac{n+1}{n} \alpha_k^* - \frac{1}{n}-W_f}}} +2\sum_{y\in[M]}W_y\\
		&= \E_{\mX_{n+1}}\lrm{\left|\cap_{k=1}^K C_k(\mX_{n+1};\alpha_k^{\ast} - \frac{1 - \alpha_k^*}{n} - W_f)\right|}+2\sum_{y\in[M]} W_y,
	\end{align*}
	where the first inequality follows from the optimality of $\widehat{\alphabold}^{\mathrm{e}}$, the second and last inequalities follow from the definition of $W_y$, the third inequality follows from \refeq{eq:quantile_error}, the third equality follows from \refeq{eq:diff_quan_infty}. 	
	Similar to \refeq{eq:bound_on_W}, we have $\E[W_y] = O(\sqrt{\log(Kn)/n})$ and consequently, by Markov’s inequality, $W_y = \Op(\sqrt{\log(Kn)/n})$.
	Combining all bounds, we conclude that
	\[\E_{\mX_{n+1}}\lrm{|\widehat{C}^{\mathrm{e}}(\mX_{n+1};\alpha)|} - \Lcal(\alphabold^* - \mW_f^\prime) \leq \Op\left( M \sqrt{\frac{\log(Kn)}{n}} \right),\]
	where $\mW_f^\prime = \vone_K (W_f + 1/n)  - \alphabold^*/n  = \Op( \sqrt{\log(K)/n}) \vone_K$ by the DKW inequality.
\end{proof}

\newpage
\section{Additional Methodological and Implementation Details}

\subsection{Intuition Behind \texorpdfstring{$\alpha$}{Alpha}-Allocation}
To further illustrate why COLA yields shorter prediction intervals than the selection-based method, we take Case 2 from the simulation study as a representative example. Fixing $n=300$ and keeping all other configurations the same as in Section~\ref{sec:numerical_marginal}, we perform a single random experiment in which COLA-e allocates $\widehat{\alphabold}^{\mathrm{e}} = (0.007,0.073,0.02)$ while EFCP selects the second score. Figure~\ref{fig:intuition} compares the prediction intervals constructed by COLA-e and EFCP for the same randomly drawn test instance, illustrating how each method forms its final prediction set.

\begin{figure}[H]
	\centering
	\includegraphics[width=0.9\linewidth]{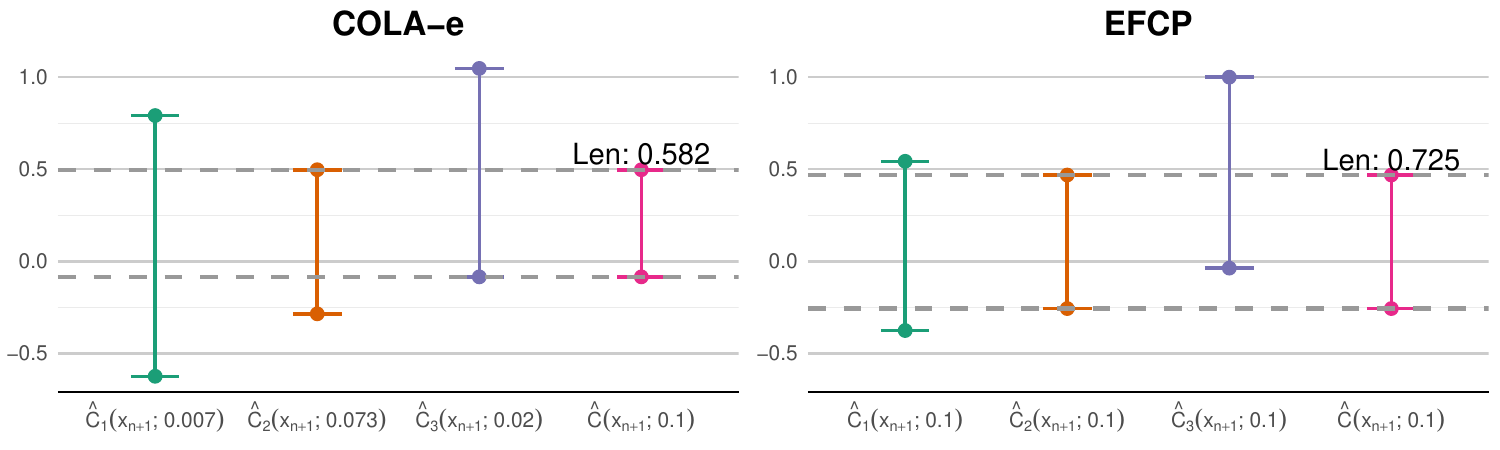}
	\caption{Prediction intervals for the same randomly drawn test instance under COLA-e and EFCP. ``Len'' denotes the length of the corresponding interval.}
	\label{fig:intuition}
\end{figure}

\subsection{Pseudocode of Stepwise Algorithm}\label{sec:stepwise}
For brevity, let $\Lcal_{n,\Scal} = \min_{\substack{\alphabold \in \boldsymbol{\Theta}_{\mathrm{grid}}(\alpha,n)\\ \mathrm{supp}(\alphabold) \subset \Scal}} \Lcal_n(\alphabold)$ denote the minimal value under support constraint.
\begin{algorithm}[H]
	\caption{Stepwise Optimization Algorithm}
	\label{alg:stepwise}
	\begin{spacing}{1}
	\begin{algorithmic}[1]
		\Require Loss function $\Lcal_n(\alphabold)$, miscoverage rate $\alpha\in(0,1)$, $k_{\max} \ge 1$, $T\ge 1$
		\Ensure $\widehat{\alphabold}$ in \refeq{eq:alpha_eff}
		
		\If{$K \le k_{\max}$} \State\Return $\widehat{\alphabold} \gets \arg\min_{\alphabold \in \boldsymbol{\Theta}_{\mathrm{grid}}(\alpha,n)} \Lcal_n(\alphabold)$
		\EndIf
		
		\State Initialize $\Scal^{(0)} \gets \varnothing$, $L^{(0)} \gets \infty$
		
		\For{$t=1$ to $T$}
		\State $(k', L^{(t)}) \gets (\arg\min_{k \in [K]} \Lcal_{n,\Scal^{(t-1)} \cup \{k\}},\min_{k \in [K]} \Lcal_{n,\Scal^{(t-1)} \cup \{k\}})$ \Comment{Forward step}
		\If{$L^{(t)} \ge L^{(t-1)}$} 
		\State \Return $\widehat{\alphabold}^{(t-1)}$
		\EndIf
		\State $\Scal^{(t)} \gets \Scal^{(t-1)} \cup \{k'\}$
		\State $\widehat{\alphabold}^{(t)} \gets \arg\min_{\substack{\alphabold \in \boldsymbol{\Theta}_{\mathrm{grid}}(\alpha,n)\\ \mathrm{supp}(\alphabold) \subset \Scal^{(t)}}} \Lcal_n(\alphabold)$
		\State $\Scal^{(t)} \gets \mathrm{supp}(\widehat{\alphabold}^{(t)})$ \Comment{Backward step}
		\If{$|\Scal^{(t)}| \ge k_{\max}$} \State\Return $\widehat{\alphabold}^{(t)}$
		\EndIf
		\EndFor
		
		\State \Return $\widehat{\alphabold}^{(T)}$
	\end{algorithmic}
	\end{spacing}
\end{algorithm}

\subsection{Summary of COLA Variants}\label{app:cola_summary}
Table~\ref{tab:cola_comparison} summarizes of the four COLA variants in terms of allocation type, coverage, validity, efficiency, and computational cost. Computation time is measured per test point.

\begin{table}[!ht]
	\centering
	\caption{Summary of the four COLA variants.}
	\label{tab:cola_comparison}
	\renewcommand{\arraystretch}{1.5}
	\resizebox{\textwidth}{!}{%
		\begin{tabular}{lccccc}
			\toprule
			\textbf{Method} & \textbf{Allocation} & \textbf{Coverage} & \textbf{Validity} & \textbf{Efficiency} & \textbf{Computation} \\
			\midrule
			$\mathrm{COLA\text{-}e}$ & Average & Marginal       & Asymptotic          & Optimal allocation (asymptotic)         & Moderate \\
			$\mathrm{COLA\text{-}s}$ & Average & Marginal       & Finite-sample       & Reduced (due to splitting)   & Low \\
			$\mathrm{COLA\text{-}f}$ & Average & Marginal       & Finite-sample       & Strong empirical efficiency  & High \\
			$\mathrm{COLA\text{-}l}$ & Individual & Conditional & Asymptotic &  Locally optimal allocation (asymptotic) & Moderate \\
			\bottomrule
		\end{tabular}%
	}
\end{table}

\section{Additional Numerical Results}\label{sec:numerical_supp}
\subsection{Additional Results on Marginal-level Performance}
In this section, the hold-out data size is fixed at $n = 300$, with all other settings kept identical to those in Section~\ref{sec:numerical_marginal}. In Section~\ref{sec:numerical_case1}, under Case 1, when a single conformal prediction set dominates, COLA exhibits performance comparable to EFCP and VFCP, naturally allocating the entire confidence budget to the set identified by EFCP and VFCP methods. Section~\ref{sec:numerical_case2}, under Case 2, compares COLA with alternative approaches, demonstrating its advantages over both score-level combination method \citep{luo2025weighted} and model-level combination method \citep{yang2024selection}. Section~\ref{subsec:comparison_s_e_f} presents a comparison of the performance of the three COLA variants: $\mathrm{COLA\text{-}e}$, $\mathrm{COLA\text{-}s}$, and $\mathrm{COLA\text{-}f}$. Finally, Section~\ref{sec:numerical_alpha} reports results across varying miscoverage levels $\alpha$. Across all scenarios, COLA-e consistently achieves the most favorable performance.

\subsubsection{Case 1: different learning algorithms}\label{sec:numerical_case1}
Figure~\ref{fig:case1} presents the comparison between COLA-e/COLA-s and EFCP/VFCP under Case 1. Both methods substantially outperform the remaining approaches, with EFCP/VFCP exhibiting performance comparable to COLA-e/COLA-s. Further analysis reveals that the score derived from the neural network dominates the others. The COLA-s allocation vector $\widehat{\alphabold}^{\mathrm{s}}$ is then compared with the VFCP selection, encoded as an $\alpha$-scaled one-hot vector $\widehat{\alphabold}^{\mathrm{VFCP}}$ of matching dimensionality. Figure~\ref{fig:alpha_com} displays the frequency distribution of the $\|\widehat{\alphabold}^{\mathrm{s}}-\widehat{\alphabold}^{\mathrm{VFCP}}\|_{\infty}$ showing that COLA-s allocates the confidence budget in a manner consistent with VFCP under this setting.

\begin{figure}[H]
	\centering
	\includegraphics[width=0.75\linewidth]{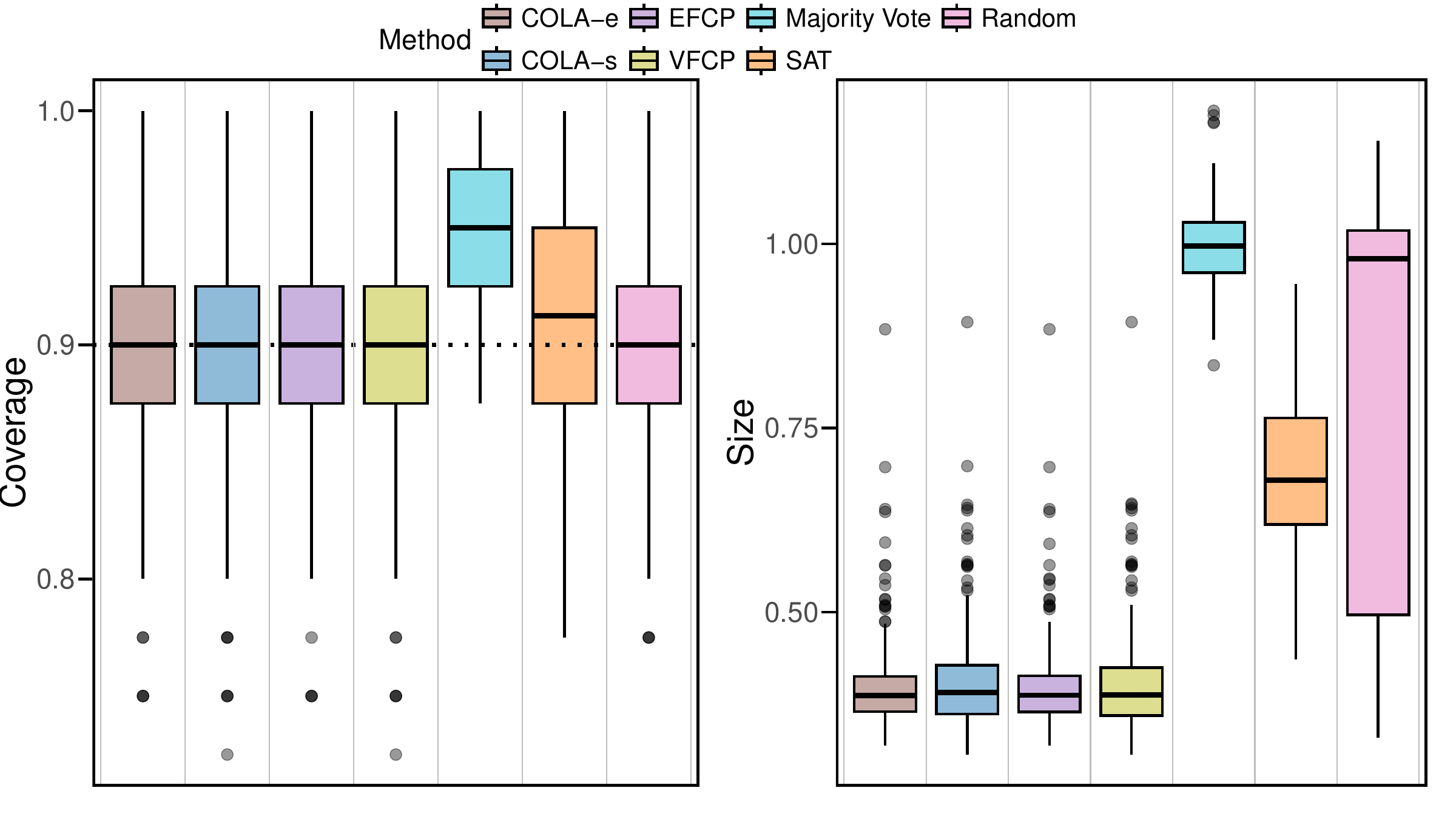}
	\caption{Coverage and size under Case 1.}
	\label{fig:case1}
\end{figure}

\begin{figure}[H]
	\centering
	\includegraphics[width=0.75\linewidth]{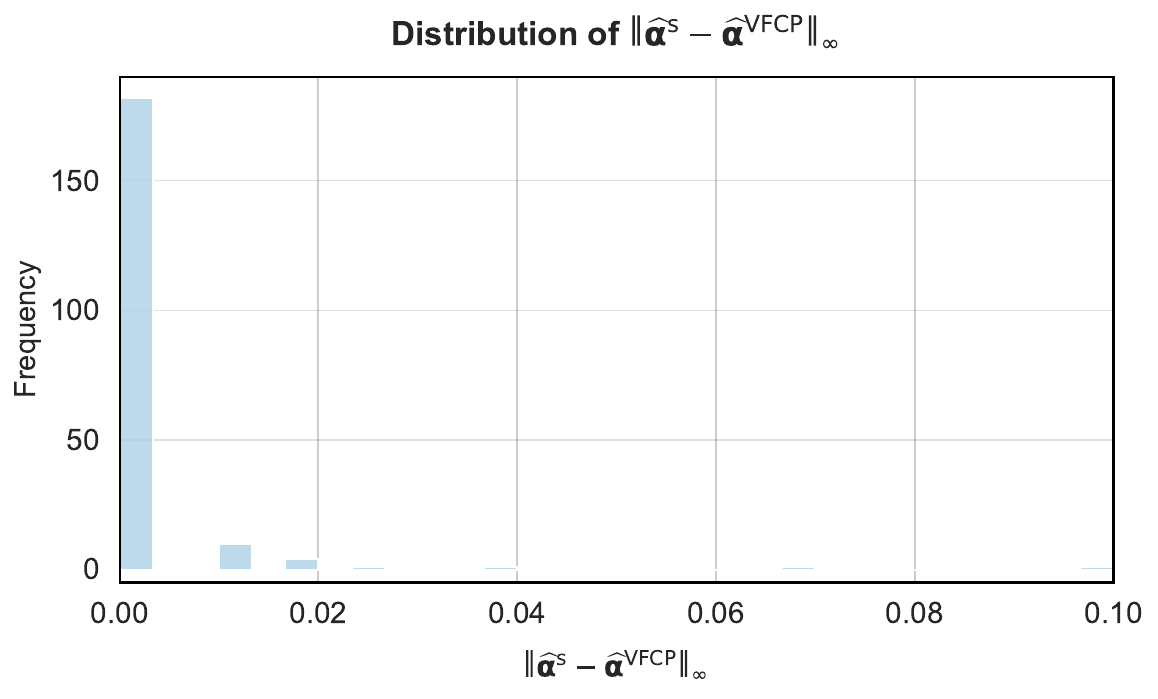}
	\caption{Histogram of $\|\widehat{\alphabold}^{\mathrm{s}}-\widehat{\alphabold}^{\mathrm{VFCP}}\|_{\infty}$ over 200 repetitions.}
	\label{fig:alpha_com}
\end{figure}

\subsubsection{Case 2: different nonconformity measures}\label{sec:numerical_case2}
In this section, under Case 2, COLA is compared with two additional methods, namely score-level combination method (LZ) \citep{luo2025weighted} and model-level combination method (VF-lin) \citep{yang2024selection}. Since VF-lin is limited to combining predictive models, its implementation here is restricted to residual scores obtained from the combination of the MARS and RF models. Figure~\ref{fig:case2} summarizes the results. COLA demonstrates superior performance over competing approaches. The VF-lin method is limited to combining conditional mean estimators, whereas the LZ method underperforms when score scales differ markedly. 

\begin{figure}[H]
	\centering
	\includegraphics[width=0.9\linewidth]{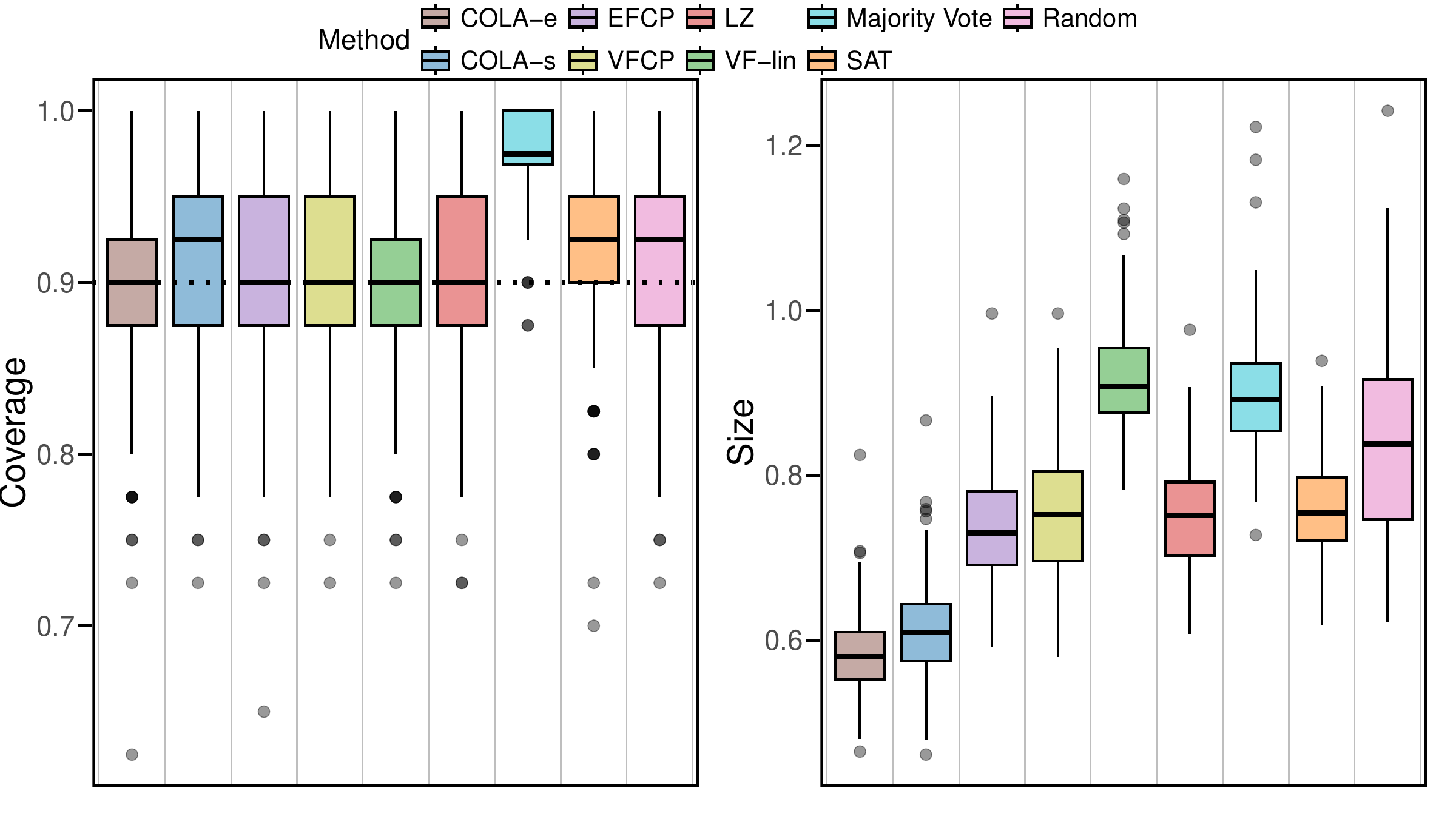}
	\caption{Coverage and size for combining different types of nonconformity scores in Case 2.}
	\label{fig:case2}
\end{figure}

\subsubsection{Comparison of COLA-e, COLA-s and COLA-f}\label{subsec:comparison_s_e_f}
This section adopts the Case 2 setting to compare the performance of $\mathrm{COLA\text{-}e}$, $\mathrm{COLA\text{-}s}$, and $\mathrm{COLA\text{-}f}$ with respect to coverage, prediction set size, and runtime, across varying $n$ and nominal coverage levels $1-\alpha$. The results are summarized in Table~\ref{tab:COLA_s_e_f}. It is observed that $\mathrm{COLA\text{-}s}$ and $\mathrm{COLA\text{-}f}$ guarantee finite-sample coverage, whereas $\mathrm{COLA\text{-}e}$ exhibits slight undercoverage. In terms of prediction set size, $\mathrm{COLA\text{-}s}$ yields the largest sets due to sample splitting. The full conformalization in $\mathrm{COLA\text{-}f}$ results in substantially higher computational cost compared to the other two methods.
\begin{table}[H]
	\caption{Coverage, size and runtime of $\mathrm{COLA\text{-}e}$, $\mathrm{COLA\text{-}s}$, and $\mathrm{COLA\text{-}f}$.}
	\label{tab:COLA_s_e_f}
	\centering
	\renewcommand{\arraystretch}{1.05}
	\setlength{\tabcolsep}{15pt}
	\resizebox{0.8\columnwidth}{!}{%
		\begin{tabular}{lllcc S[table-format=4.2]}
			\hline
			$\alpha$ & $n$ & method & coverage (\%) & size & {runtime (sec)} \\ \hline
			\multirow{6}{*}{0.05} & \multirow{3}{*}{100} & $\mathrm{COLA\text{-}e}$ & 93.85 & 0.70 & 0.05 \\ 
			&  & $\mathrm{COLA\text{-}s}$ & 95.45 & 0.78 & 0.06 \\
			&  & $\mathrm{COLA\text{-}f}$ & 95.55 & 0.73 & 71.45 \\ \cline{2-6} 
			& \multirow{3}{*}{300} & $\mathrm{COLA\text{-}e}$ & 94.65 & 0.69 & 0.23 \\ 
			&  & $\mathrm{COLA\text{-}s}$ & 95.85 & 0.72 & 0.07 \\
			&  & $\mathrm{COLA\text{-}f}$ & 95.40 & 0.71 & 498.87 \\ \hline
			\multirow{6}{*}{0.10} & \multirow{3}{*}{100} & $\mathrm{COLA\text{-}e}$ & 89.30 & 0.59 & 0.11 \\ 
			&  & $\mathrm{COLA\text{-}s}$ & 91.60 & 0.63 & 0.04 \\
			&  & $\mathrm{COLA\text{-}f}$ & 91.20 & 0.62 & 205.85 \\ \cline{2-6} 
			& \multirow{3}{*}{300} & $\mathrm{COLA\text{-}e}$ & 88.70 & 0.59 & 0.89 \\ 
			&  & $\mathrm{COLA\text{-}s}$ & 90.30 & 0.62 & 0.26 \\
			&  & $\mathrm{COLA\text{-}f}$ & 90.15 & 0.61 & 1793.23 \\ \hline
		\end{tabular}%
	}
\end{table}

\subsubsection{Varying \texorpdfstring{$\alpha$}{alpha}}\label{sec:numerical_alpha}
This section presents results for varying miscoverage levels $\alpha$ from 0.025 to 0.02, as shown in Figure~\ref{fig:lineplot_a}.

\begin{figure}[t!]
	\centering
	\begin{minipage}{0.8\textwidth}
		\centering
		\includegraphics[trim=4.9cm 14.4cm 4.9cm 0cm, clip, height=0.4cm]{lineplot_K.pdf}
	\end{minipage}
	\vspace{0.1em}
	\begin{minipage}{0.32\textwidth}
		\includegraphics[trim=1cm 0 7.3cm 0, clip, height=9cm]{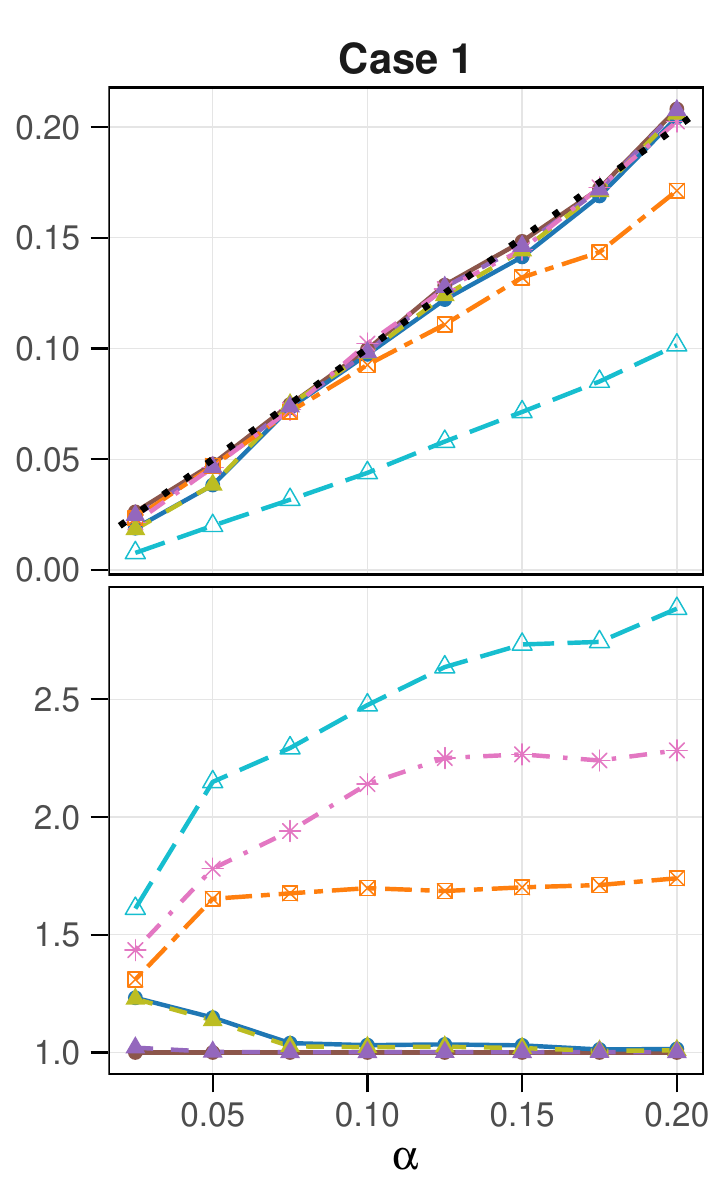}
	\end{minipage}%
	\hfill
	\begin{minipage}{0.32\textwidth}
		\includegraphics[trim=1cm 0 7.3cm 0, clip, height=9cm]{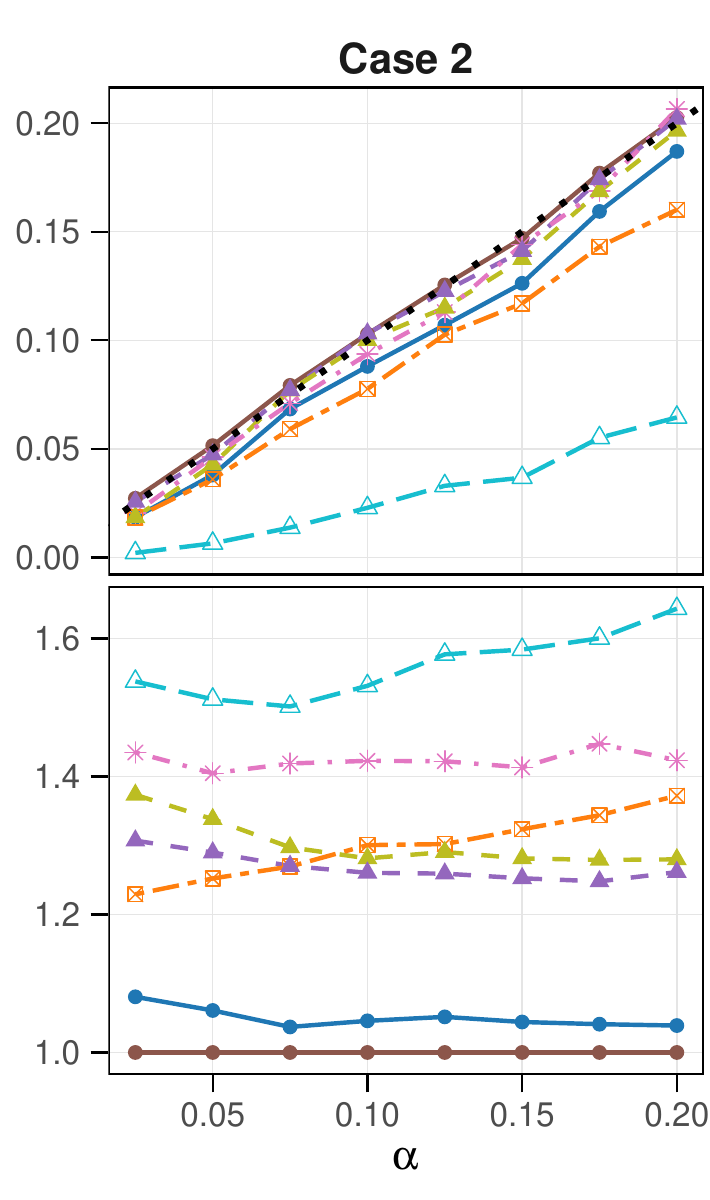}
	\end{minipage}%
	\hfill
	\begin{minipage}{0.34\textwidth}
		\includegraphics[trim=1cm 0 5.5cm 0, clip, height=9cm]{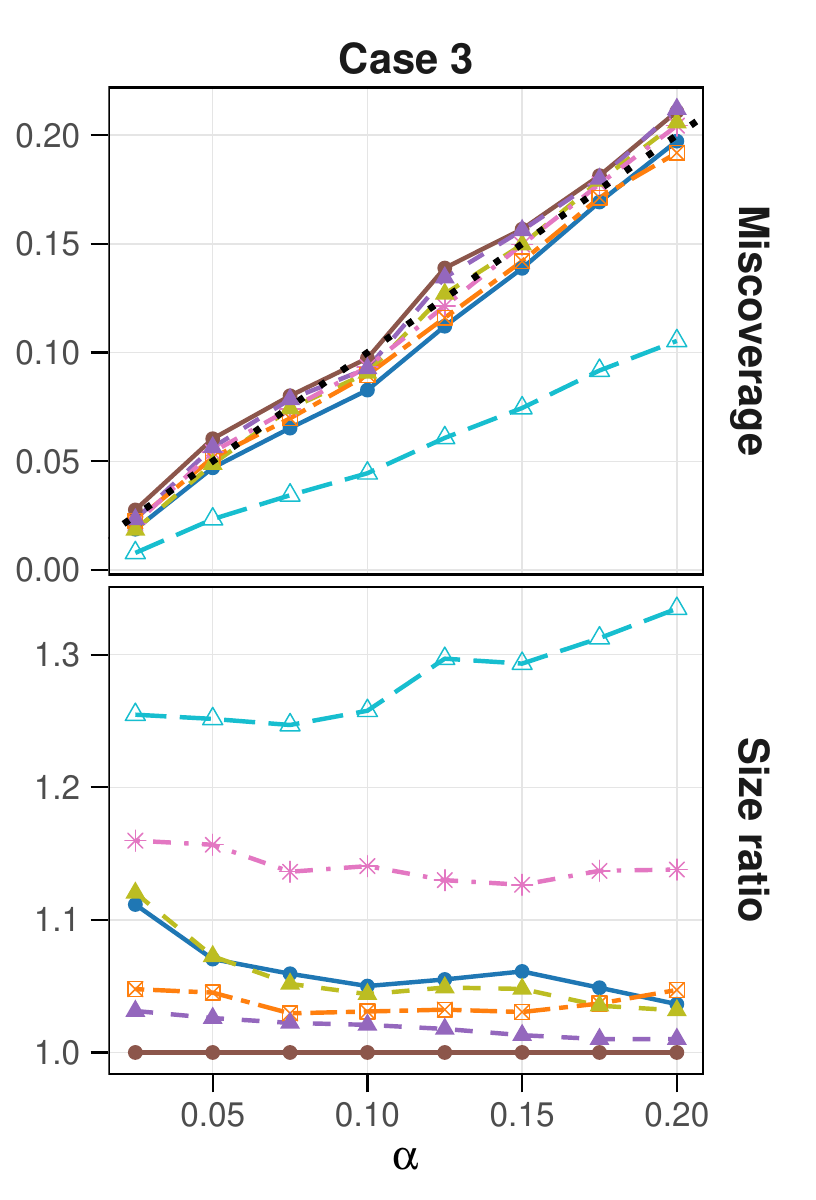}
	\end{minipage}
	\caption{Performance under Cases 1–3 with varying $\alpha$. Top row: coverage; bottom row: prediction set size ratio relative to COLA-e. The black dotted line represents the nominal miscoverage level.}
	\label{fig:lineplot_a}
\end{figure}

\subsection{Smoothing and Stepwise Optimization}\label{sec:optimization_comparision}
\subsubsection{Smoothing Optimization}
Due to the non-smooth nature of $\Lcal_n(\alphabold)$, a standard approach is to approximate its non-smooth components with smooth surrogates (e.g., \cite{xie2024boosted,kiyani2024length}). Suppose $K$ predictive models $\widehat{\mu}_k(\vx)$ are used to construct residual-score-based conformal prediction sets of the form $[\widehat{\mu}_k(\vx) - \widehat{q}_{k,\alpha_k}, \widehat{\mu}_k(\vx) + \widehat{q}_{k,\alpha_k}]$ with $\widehat{q}_{k,\alpha_k} = Q_{\alpha_k}(\{S_{k,i} = \abs{Y_i - \widehat{\mu}_k(\mX_i)}\}_{i \in [n]})$. Other types of interval-valued prediction sets can be smoothed analogously. For any input $\vx$ and allocation $\alphabold$, the prediction set is given by
$[\max_k \{\widehat{\mu}_k(\vx) - \widehat{q}_{k,\alpha_k}\}, \min_k \{\widehat{\mu}_k(\vx) + \widehat{q}_{k,\alpha_k}\}],$
where the non-smoothness arises from operations such as $\min$, $\max$ and indicator functions in empirical quantile definitions.

To approximate $\max$ and $\min$, the log-sum-exp function \citep{boyd2004convex} is used, yielding
\begin{align*}
	\widetilde{\max}(\{v_k\}_{k \in [K]})  = \tau_1^{-1} \log(\sum_{k\in[K]} \exp( \tau_1 v_k) ), ~~~~
	\widetilde{\min}(\{v_k\}_{k \in [K]})  = -\tau_1^{-1} \log(\sum_{k\in[K]} \exp( -\tau_1 v_k) ),
\end{align*}
with smoothing parameter $\tau_1 = 20$ in experiments.

For the quantiles, the empirical distribution is first smoothed as $\widetilde{F}_k(s) := \frac{1}{n} \sum_{i\in[n]} \Phi(\frac{s - S_{k,i}}{\tau_{2,k}})$, where $\Phi$ is the standard normal CDF and $\tau_{2,k} = \widehat{\sigma}_k n^{-1/5}$ is selected via Silverman's rule \citep{silverman2018density}, with $\widehat{\sigma}_k$ the sample standard deviation of $\{S_{k,i}: i \in [n]\}$. The smoothed quantile is defined as the inverse, $\widetilde{q}_{k,\alpha_k} = \widetilde{F}_k^{-1}(\alpha_k)$, computed numerically via binary search. Its gradient is
\[\nabla_{\alpha_k} \widetilde{q}_{k,\alpha_k} = n \tau_{2,k} / \sum_{i\in[n]} \phi(\frac{s - S_{k,i}}{\tau_{2,k}}),\]
where $\phi$ is the standard normal density.

Finally, the smoothed objective is
\[\widetilde{\Lcal}_n(\alphabold):= \frac{1}{n} \sum_{i\in[n]} \Big(\widetilde{\min}\{\widehat{\mu}_k(\mX_i) + \widetilde{q}_{k,\alpha_k}\} - \widetilde{\max}\{\widehat{\mu}_k(\mX_i) - \widetilde{q}_{k,\alpha_k}\}\Big) ,\]
whose gradient can be computed via the chain rule, thereby enabling gradient-based optimization.

\subsubsection{Numerical Comparison}
The smoothing and stepwise approaches are compared in the Case~3 setup in Section~\ref{sec:numerical_marginal}, with $K = 5$ or $K = 100$. The empirical loss $\Lcal_n(\alphabold)$ of the allocations produced by both methods is reported in Figure~\ref{fig:optimization_comparison}. For $K = 5$, grid search is computationally feasible and identifies the global optimum, which is closely approximated by the stepwise solution. The stepwise approach consistently achieves lower empirical loss than the smooth method, with the difference increasing for $K = 100$.
\begin{figure}[H]
	\centering
	\includegraphics[width=0.7\linewidth]{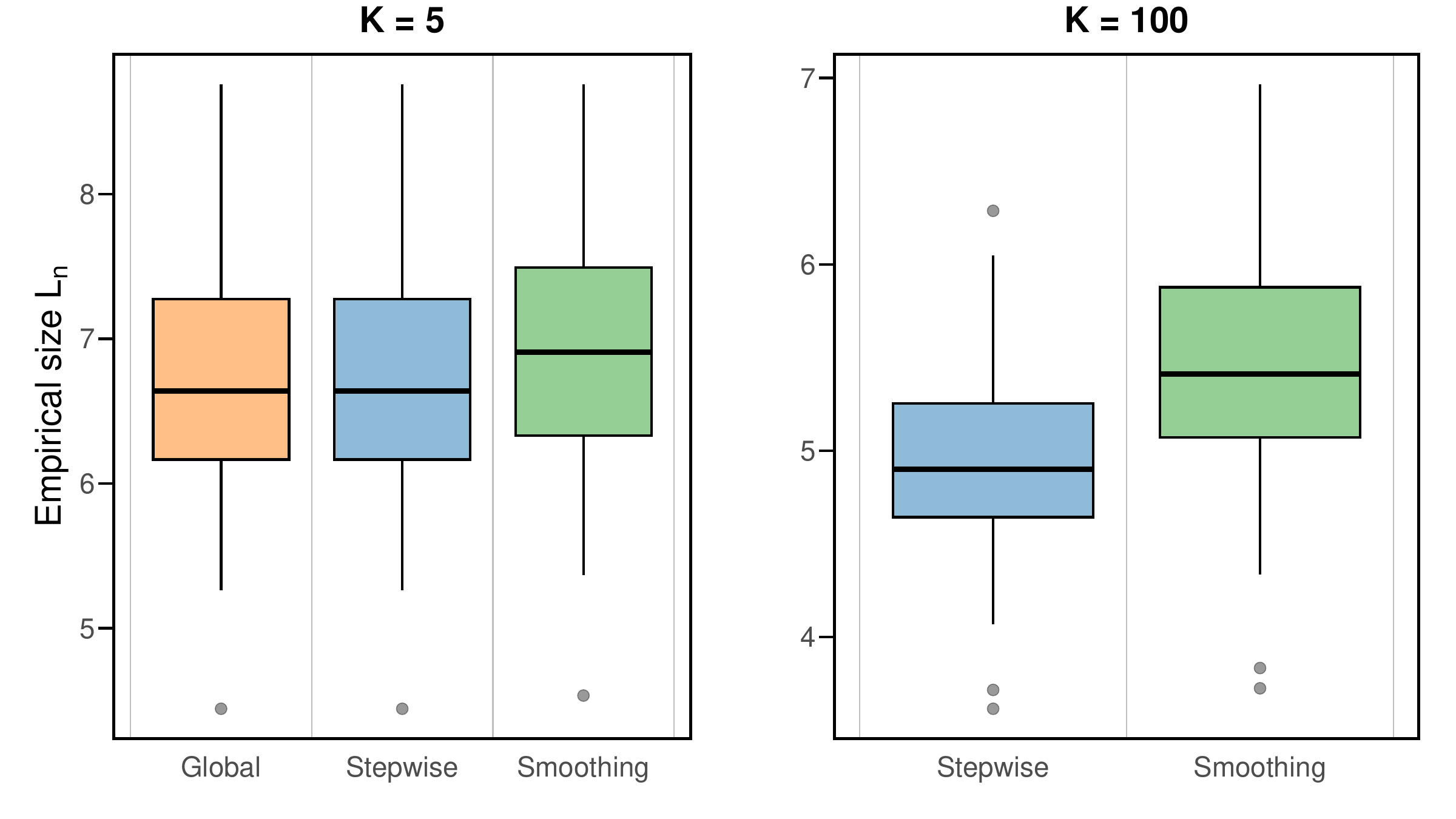}
	\caption{Empirical size $\Lcal_n(\alphabold)$ of the allocations obtained via smoothing and stepwise optimization.}
	\label{fig:optimization_comparison}
\end{figure}

\subsection{Additional Results on Real Datasets}\label{app:real_data}
The results on three additional UCI datasets---\textit{News Popularity, Kernel Performance, and Protein Structure}---are summarized in Table~\ref{tab:uci_remaining}. 
\begin{table}[H]
	\centering
	\caption{Coverage and size of prediction sets on three additional UCI datasets ($\alpha = 0.1$).}
	\label{tab:uci_remaining}
	\begin{tabular}{
			l
			>{\centering\arraybackslash}p{1.5cm}
			>{\centering\arraybackslash}p{1.5cm}
			>{\centering\arraybackslash}p{1.5cm}
			|
			>{\centering\arraybackslash}p{1.5cm}
			>{\centering\arraybackslash}p{1.5cm}
			>{\centering\arraybackslash}p{1.5cm}
		}
		\toprule
		\multirow{2}{*}{\textbf{Method}} 
		& \multicolumn{3}{c|}{\textbf{Coverage (\%)}} 
		& \multicolumn{3}{c}{\textbf{Size}} \\
		& \textit{Protein} & \textit{News} & \textit{Kernel} 
		& \textit{Protein} & \textit{News} & \textit{Kernel} \\
		\midrule
		$\mathrm{COLA\text{-}e}$ & 90.88 & 92.00 & 91.46   & 14.50 & 9570.65  & 372.28 \\
		$\mathrm{COLA\text{-}s}$       & 91.34 & 92.15 & 92.15   & 15.08 & 11429.48 & 407.67 \\
		EFCP    & 90.20 & 90.03 & 90.54   & 15.41 & 10688.04 & 385.58 \\
		VFCP         & 89.93 & 90.07 & 90.76   & 15.60 & 10981.36 & 390.40 \\
		Majority Vote  & 96.55 & 96.54 & 96.94   & 19.08 & 20394.38 & 670.15 \\
		SAT & 91.18 & 92.73 & 95.62   & 14.63 & 15168.10 & 509.44 \\
		Random     & 90.54 & 90.34 & 90.25   & 16.26 & 14995.35 & 606.19 \\
		\bottomrule
	\end{tabular}
\end{table}

\subsection{Description of Real Datasets}
\begin{table}[H]
	\centering
	\caption{The six UCI datasets used in our experiments. 
		$n$ denotes the total number of samples, and $d$ is the feature dimension.}
	\label{tab:dataset_info}
	\begin{tabular}{lccp{11cm}}
		\toprule
		\textbf{Dataset} & \textbf{$n$} & \textbf{$d$} & \textbf{URL (http://archive.ics.uci.edu/ml/datasets/*)} \\
		\midrule
		Blog & 52,397 & 280 & BlogFeedback \\
		Concrete & 1,030 & 8 & Concrete+Compressive+Strength \\
		Conduct & 21,263 & 81 & Superconductivty+Data \\
		Protein & 45,730 & 8 & Physicochemical+Properties+of+Protein+Tertiary+Structure\\
		News & 39,644 & 59 & Online+News+Popularity \\
		Kernel & 241,600 & 14 & SGEMM+GPU+kernel+performance \\
		\bottomrule
	\end{tabular}
\end{table}

\subsection{Predictive Models Implemented in R}
The predictive models employed in the experiments—ridge regression, regression trees, random forest, neural network, multivariate adaptive regression splines, and quantile regression forest—are implemented using the respective \texttt{R} packages: \texttt{glmnet}, \texttt{tree}, \texttt{randomForest}, \texttt{nnet}, \texttt{earth}, and \texttt{quantregForest}.

\end{document}